\documentclass[
    reprint,
    amsmath,amssymb,
    nofootinbib,
    superscriptaddress,
    aps,
    noeprint,
    longbibliography,
    prl
]{revtex4-2}

\usepackage{amsthm}
\usepackage{braket}
\usepackage{graphicx}
\usepackage{placeins}
\usepackage[dvipsnames]{xcolor}
\usepackage{setspace} 
\usepackage{algorithm}
\usepackage{algpseudocode}

\usepackage[english]{babel}
\babelprovide[import=en]{en}
\babelprovide[import=en-GB]{en-GB}

\usepackage{hyperref}
\hypersetup{
    pdfstartview={FitH},
    colorlinks=true,
    linkcolor=blue,
    citecolor=blue,
    filecolor=magenta,
    urlcolor=blue
}

\theoremstyle{plain}
\newtheorem{theorem}{Theorem}[section]
\newtheorem{lemma}[theorem]{Lemma}
\newtheorem{corollary}[theorem]{Corollary}
\newtheorem{definition}[theorem]{Definition}

\DeclareMathOperator*{\E}{\mathbb{E}}
\begin{document}

\title{The practical cost of magic state cultivation
}

\author{Rohan~Mehta}
\email{rohanmehta@g.harvard.edu}
\affiliation{Department~of~Physics,~Harvard~University,~Cambridge,~MA~02138,~USA}
\author{Varun~Menon}
\affiliation{Department~of~Physics,~Harvard~University,~Cambridge,~MA~02138,~USA}
\author{Hengyun~Zhou}
\affiliation{Department of Electrical Engineering and Computer Science, Massachusetts Institute of
Technology, Cambridge, MA 02139, USA}
\author{Mikhail~D.~Lukin}
\affiliation{Department~of~Physics,~Harvard~University,~Cambridge,~MA~02138,~USA}
\author{J. Pablo Bonilla Ataides}
\email{jbonillaataides@g.harvard.edu}
\affiliation{Department~of~Physics,~Harvard~University,~Cambridge,~MA~02138,~USA}

\begin{abstract}
Magic state cultivation is a promising method for preparing high-fidelity non-Clifford resource states—a central component of fault-tolerant quantum computing—with low spacetime overhead. In benchmarking cultivation, existing resource estimates assume access to noiseless transversal measurement of the output magic state to allow for postselection during the final “escape” stage of the cultivation process. However, such measurements would destroy the magic state and are therefore unavailable during computation. We develop an open-boundary postselection method, Caliper, that uses only nondestructive mid-circuit syndrome information and substantially outperforms existing methods operating under these constraints. In some regimes, our postselection method uniquely recovers the resource efficiency of prior estimates. However, in other cases, we observe a tradeoff. At spacetime costs reported in previous works, the achievable logical error rates can be orders of magnitude higher, and recovering the reported logical error rate requires substantially larger codes and longer escape stages. Our results highlight how practical considerations involving mid-circuit postselection can impact the design and performance of magic state cultivation protocols.
\end{abstract}

\maketitle

Quantum error correction (QEC) is believed to be critical for reliable quantum computation at large
scales \cite{shor_scheme_1995, steane_error_1996, knill_resilient_1998,
aharonov_fault-tolerant_2008, kitaev_fault-tolerant_2003, terhal_quantum_2015, fowler_surface_2012}.
Recent experiments have demonstrated QEC on neutral-atom processors
\cite{bluvstein_logical_2024, bluvstein_fault-tolerant_2026, reichardt_fault-tolerant_2025, computing2026quantumerrorcorrectiontoric},
superconducting chips \cite{acharya_suppressing_2023, acharya_quantum_2025, krinner_realizing_2022},
and other platforms \cite{ryan-anderson_realization_2021, postler_demonstration_2022,
paetznick_demonstration_2024, abobeih_fault-tolerant_2022}.
Performing computation on encoded information, however, also requires fault-tolerant non-Clifford
operations, which are necessary for universal quantum computation \cite{gottesman_heisenberg_1998, knill_fault-tolerant_2004,
bravyi_universal_2005, eastin_restrictions_2009, zeng_transversality_2011}. 

Magic state cultivation \cite{chamberland_very_2020, itogawa_efficient_2025, gidney_magic_2024,
rosenfeld_magic_2025, sahay2026fold, vaknin_high_2026, claes_cultivating_2025, chen_efficient_2025,
hirano2026efficient} has recently emerged as a leading approach for preparing magic states with low
logical error rates, which can then be consumed to implement non-Clifford gates
\cite{gottesman_demonstrating_1999, zhou_methodology_2000, knill_fault-tolerant_2004,bravyi_universal_2005,
litinski_magic_2019}. Relatively noisy magic states begin in small codes \cite{li_magic_2015, gidney_cleaner_2023} and are continuously
refined through measurements of Clifford stabilizers and postselection \cite{chamberland_fault-tolerant_2019, chamberland_very_2020, itogawa_efficient_2025}. As the post-selected logical
error rates decrease, the code must be rapidly enlarged to protect the high-fidelity state in the
so-called ``escape'' stage. As this interface is noisy, further postselection is required,
implying repetitions that inflate the expected spacetime volume. The escape stage is the most
expensive component of cultivation \cite{sahay2026fold}, so its postselection overhead is critical
to the viability of the protocol.

Recent work has addressed this challenge by using sliding-scale postselection
\cite{bombin_fault-tolerant_2024, smith_mitigating_2024,
english_thresholds_2025, meister_efficient_2024, lee_efficient_2026, xie_simple_2026,
wills_forced_2026}, achieving a tunable balance between post-selected error rate and acceptance
fraction. One such method, the \emph{complementary gap} \cite{gidney_yoked_2023,
bombin_fault-tolerant_2024}, can achieve exponential error suppression at
high acceptance fractions \cite{bombin_fault-tolerant_2024, smith_mitigating_2024,english_thresholds_2025,
chen2026scalable}. It has therefore been employed across Refs.~\cite{gidney_magic_2024,
sahay2026fold, vaknin_high_2026, claes_cultivating_2025, hirano2026efficient} to facilitate a
low-overhead escape stage. However, its use as a diagnostic requires a closed temporal boundary from a terminal measurement at the end of a circuit to yield a useful result.
Applied directly to a mid-circuit observable, the complementary gap collapses into a narrow band whose width is independent of code distance, obscuring the differences in reliability needed for postselection
(Supplementary Material). Existing simulations close this boundary using information from a
noiseless transversal measurement of the output codeblock. In practice, this measurement would
destroy the magic state, so the decoder must instead decide whether to retain it using only
nondestructive syndrome information.

\begin{figure*}
    \centering
    \includegraphics[width=\textwidth]{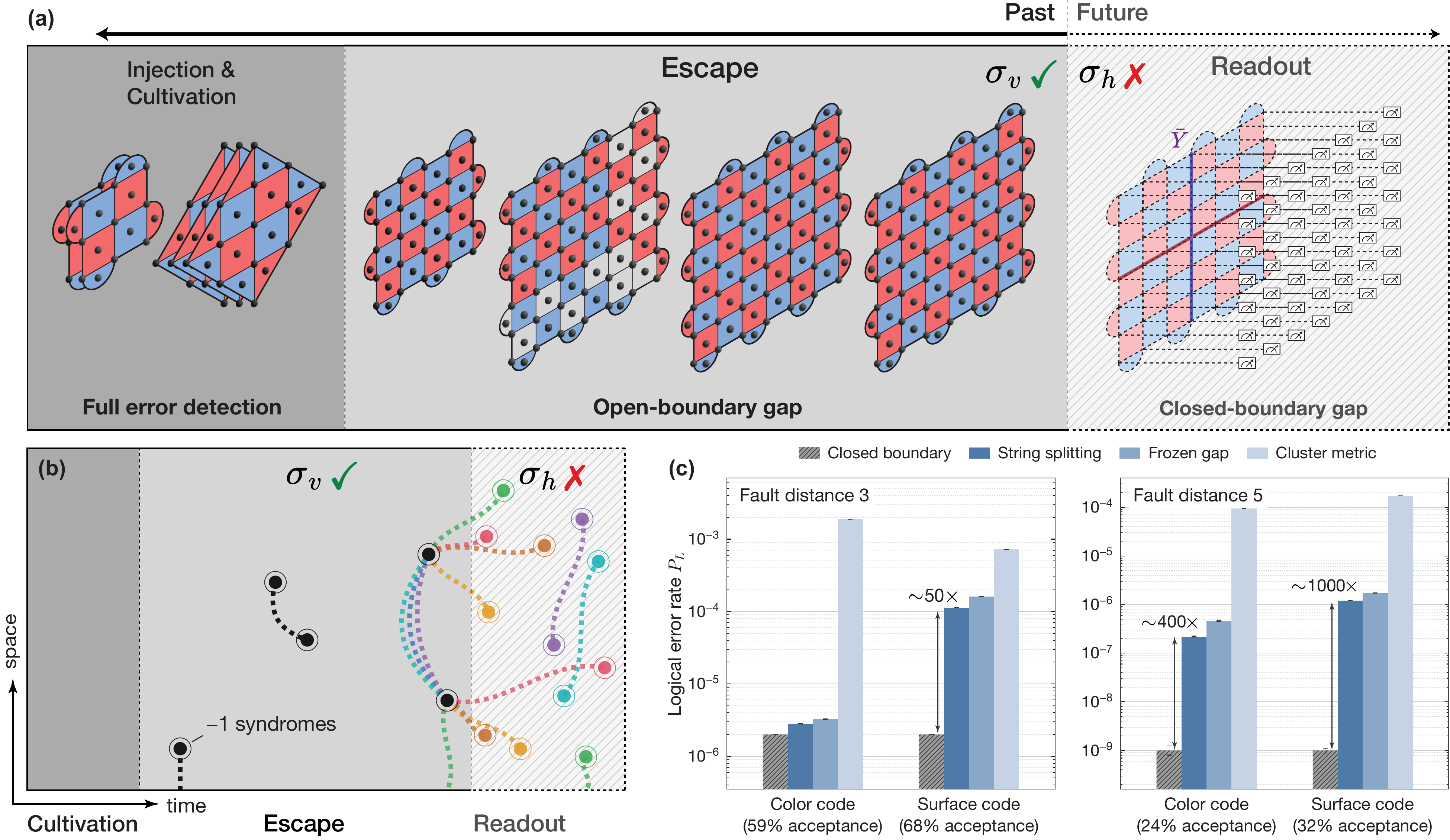}
    \caption{\textbf{Magic state cultivation with mid-circuit postselection.} (a) Stages
    of the fold-transversal cultivation protocol on the surface code from Ref. \cite{sahay2026fold},
    which targets a logical error rate of $2 \times 10^{-6}$. Full error detection is used in the
    injection and cultivation stages (left). The escape stage (center) involves growing to a $d = 7$
    surface code and repeating $r= 3$ rounds of error correction. Readout (right) is depicted by a
    transversal measurement, though it is more accurately described by a noise-free QEC round, which
    provides even more information. As readout cannot be performed in practice, access to $\sigma_h$
    is unphysical; the gap must be predicted from $\sigma_v$ alone. (b) A cross-sectional view of
    open-boundary postselection with one spatial and one temporal dimension. Due to the hidden syndrome layer, the eventual correction, along with what constitutes a logically inequivalent correction for the complementary gap, is hard to predict. (c) Comparison of logical error rates for open- and closed-boundary
    postselection, evaluated at the acceptance fractions reported in Refs. \cite{sahay2026fold,
    gidney_magic_2024} to achieve target logical error rates (LERs) of $2\times10^{-6}$ and
    $10^{-9}$, respectively.}
    \label{fig:overview}
\end{figure*}

This creates two challenges: without the readout, less information is available, and extracting a
reliable postselection metric from the remaining syndrome can be computationally expensive. These
considerations motivate studying escape-stage postselection under realistic, mid-circuit conditions,
where the magic state is preserved rather than measured. We refer to this setting as
\emph{open-boundary postselection}.

In this Letter, we develop an efficient method, \emph{Caliper}, for postselection
using only the syndrome information available at decision time, with latency comparable to computing
the complementary gap. We find that existing efficient heuristics
\cite{lee_efficient_2026, staples_scalable_2026, birchall_macromux_2026} can fall orders of magnitude
short of the performance accessible from the same information. This deficiency can instead be
compensated by expending additional quantum resources, revealing a tradeoff between classical
inference and quantum overhead. Caliper substantially improves this tradeoff. Nevertheless, the
consequences remain strongly protocol dependent: while some cultivation resource estimates are
nearly unchanged, others require substantially larger codes and longer escape stages to approach
their reported logical error rates. The resulting increase in spacetime cost per accepted magic
state directly raises the resources required for non-Clifford gates, and hence for fault-tolerant algorithms whose
cost is dominated by magic state preparation.

\textit{Open-boundary postselection.---} For the cultivation circuits considered here,
minimum-weight perfect matching (MWPM) approximates maximum-likelihood decoding by retaining the
most likely error configuration in each logical class rather than summing over configurations
\cite{dennis_topological_2002, higgott_pymatching_2022, higgott_sparse_2025}. For a complete syndrome $\sigma$, let $w_0$ be the minimum weight of any error configuration consistent with the observed syndrome, and let $w_1$ denote the minimum weight of any complementary error, which differs from the minimum-weight error by a nontrivial logical operator. Both can be computed via matching \cite{hutter_efficient_2014}. The
complementary gap is defined as $G=w_1-w_0$ \cite{gidney_yoked_2023, bombin_fault-tolerant_2024}, which approximates
the conditional logical error probability $P_{\mathrm L}(\sigma)\approx e^{-G}/(1+e^{-G})$. After computing the complementary gap, a decoder makes a postselection decision on whether to accept or discard the shot based on whether this gap $G$ exceeds a tunable threshold $G^*$.

At decision time, however, the decoder only has access to a visible syndrome history $\sigma_v$. The final
readout would reveal an additional, hidden (future) syndrome $\sigma_h$, giving the complete syndrome
$\sigma=\sigma_v\oplus\sigma_h$~\cite{staples_scalable_2026}, but this information is inaccessible in practice. Instead, we average over the potential hidden syndromes, weighting each by how likely it is based on $\sigma_v$: $\tilde P_L(\sigma_v)=\sum_{\hat\sigma_h}P(\hat\sigma_h\mid\sigma_v)P_{\mathrm
L}(\sigma_v\oplus\hat\sigma_h)$.

Replacing each logical class by its minimum-weight error configuration gives
\begin{equation}
    \label{eq:decoding-failure-likelihood}
    \hat P_L(\sigma_v)
    = \frac{\sum_{\hat\sigma_h}e^{-w_1(\sigma_v\oplus\hat\sigma_h)}}
    {\sum_{\hat\sigma_h}\left[e^{-w_0(\sigma_v\oplus\hat\sigma_h)}
    +e^{-w_1(\sigma_v\oplus\hat\sigma_h)}\right]}.
\end{equation}
Evaluating
Eq.~\eqref{eq:decoding-failure-likelihood} requires a search over exponentially many hidden
syndromes, making it challenging to identify, at low computational cost, the
hidden syndromes that contribute most strongly to the estimated logical failure probability.

\begin{figure*}
    \centering
    \includegraphics[width=\textwidth]{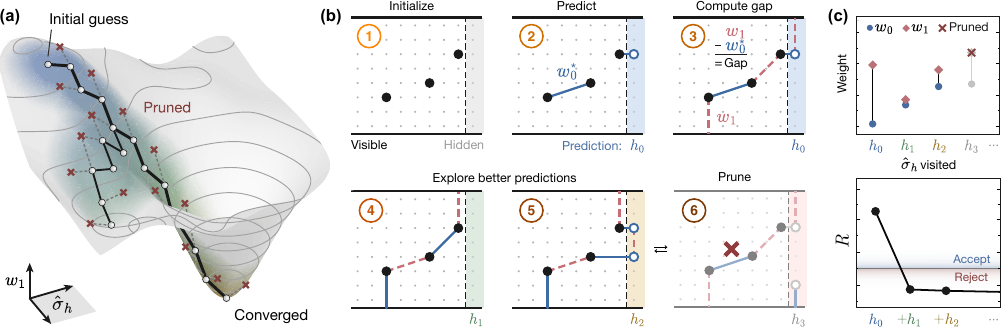}
    \caption{\textbf{Overview of Caliper}. (a) Caliper builds on top of the
    complementary gap, which when applied to the open-boundary setting returns an initial guess for
    the hidden syndrome $\hat{\sigma}_{h_0}$ and the gap. It then traverses the landscape of other
    possible hidden syndromes $\hat \sigma_h$ to find low-weight complementary corrections, $w_1$,
    pruning locally suboptimal directions from the search frontier. (b) A step-by-step walkthrough
    of this process. In steps 1-3, Caliper finds $\hat \sigma_{h_0}$ and its first guess at the gap,
    $w_1 - w_0^*$. It then iteratively explores the syndrome space (steps 4-6), preferentially expanding
    predictions with low $w_1$, until the frontier is exhausted or the expansion limit is reached.
    (c) The weights of the corrections for each prediction (top) contribute to the score
    $R$ (bottom). A predicted hidden syndrome with nearly tied corrections can lower $R$ below the acceptance
    threshold; further predictions continue to adjust the sums.}
    \label{fig:caliper}
\end{figure*}

\textit{Caliper.---} Caliper is a soft-output decoder that
approximates the likelihood of an eventual logical error with low latency. Caliper approximates
Eq.~\eqref{eq:decoding-failure-likelihood} over a set $\mathcal S$ of hidden syndromes selected by
its search, and reports the open-boundary postselection score $R \approx -\log \hat P_L^{(\mathcal S)}$.
Large $R$ indicates a reliable shot, which is accepted when $R\geq R^*$. This plays the role of $G$ for open-boundary postselection, but requires searching for the hidden syndromes that dominate the estimated logical error probability and combining their contributions.

We illustrate the key idea in Fig.~\ref{fig:caliper}. Caliper begins by decoding only the visible
syndrome, with hidden detectors connected to the boundary at zero weight. The decoder finds a minimum-weight
correction whose endpoints in the hidden layer define an initial prediction $\hat\sigma_{h_0}$, with
weight $w_0^*=w_0(\hat\sigma_{h_0})$. Fixing this prediction and decoding the opposing logical
class gives $w_1(\hat\sigma_{h_0})$. These two calls to MWPM determine the complementary gap
$G_{h_0}=w_1(\hat\sigma_{h_0})-w_0^*$ if the hidden syndrome were indeed $\hat{\sigma}_{h_0}$, 
but it does not account for other possible outcomes, whose importance depends both on how likely they are and on the probability of a
logical error if they occur. Equation~\eqref{eq:decoding-failure-likelihood} combines these two
considerations: for a fixed visible syndrome, each predicted hidden syndrome contributes to the estimated logical
error probability in proportion to $e^{-w_1(\hat\sigma_h)}$.

Our key observation is that the structure of minimum-weight matching allows this space to be
explored efficiently. Moving a $-1$ syndrome from $u$ to $v$ changes the optimal correction
weight by at most the shortest-path distance between them, accounting for the logical parity of the
path. The required graph distances
can all be precomputed once for the circuit (Supplementary Material), immediately bounding the weights for nearby syndromes.

Furthermore, flipping one hidden detector requires changing the minimum-weight correction only along a path starting at that detector. Caliper can therefore search for the new correction using precomputed distances, calling MWPM only when needed to resolve the smallest weights. This enables Caliper to explore the neighborhood of a syndrome without decoding every prediction separately.

\begin{figure*}[t!]
    \centering
    \includegraphics[width=\textwidth]{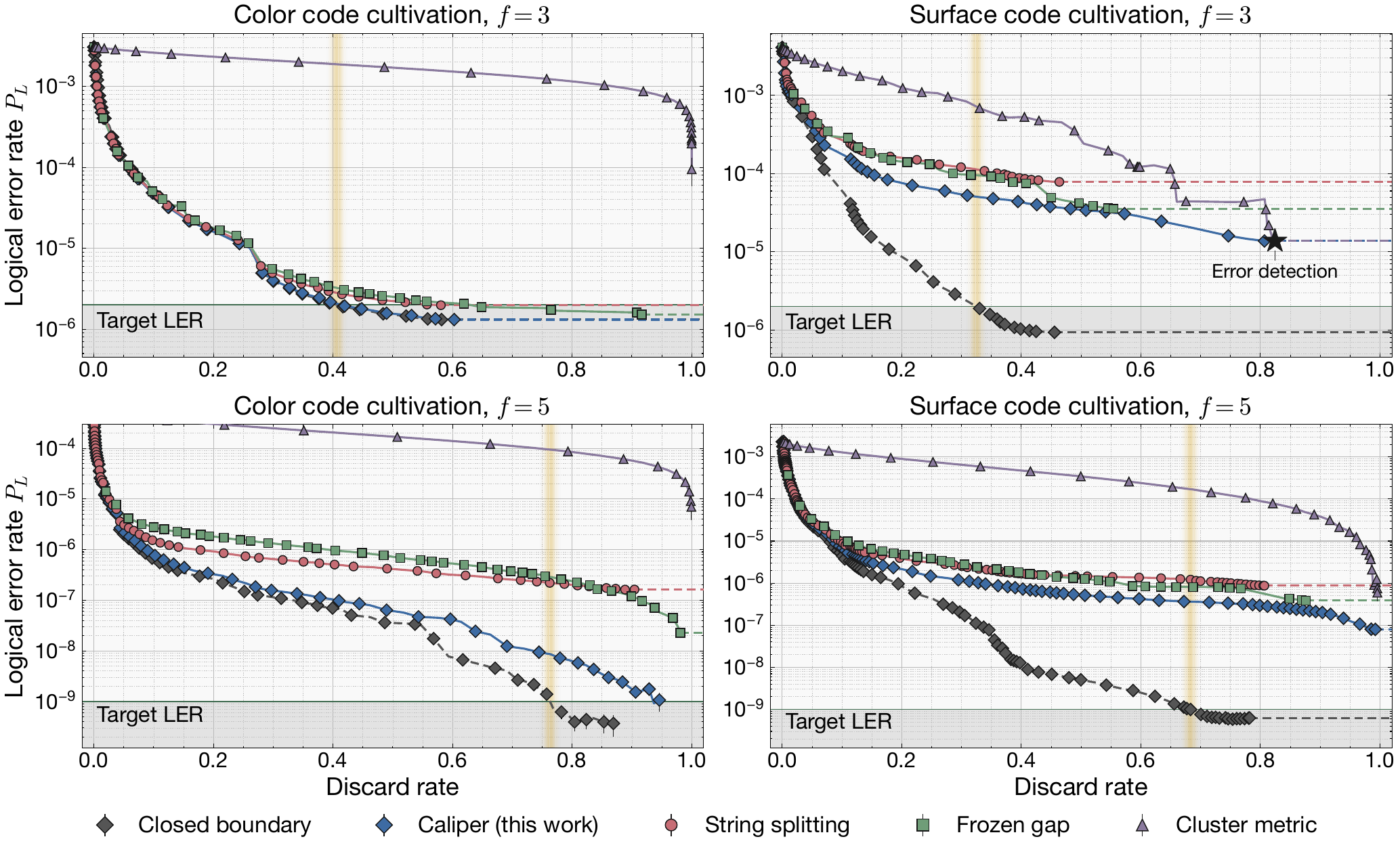}
    \caption{\textbf{Achievable suppression of logical error rate with open-boundary postselection.}
    We compare the logical error rate achieved under closed-boundary postselection at a tunable
    discard rate for the escape stage (gray) with those achieved by various open-boundary postselection protocols,
    including Caliper (blue), string splitting (red) \cite{staples_scalable_2026}, the frozen gap
    (green) \cite{birchall_macromux_2026,note_frozen_gap_parameter}, and cluster-based confidence (purple)
    \cite{lee_efficient_2026,note_cluster_confidence_metric} for the protocols in Refs. \cite{gidney_magic_2024, sahay2026fold, scc_fd5_simulation}. Target LERs of $2 \times 10^{-6}$ for fault distance 3 ($f=3$) cultivation and $10^{-9}$ for
    distance 5 ($f=5$) cultivation are annotated. Our reproduction of closed-boundary postselection matches the data in Refs. \cite{gidney_magic_2024, sahay2026fold} to within error bars, and the discard rate at which these curves reach the target LER is highlighted. Dashed horizontal lines indicate the LER plateaus reached when the postselection metrics
    saturate. At these points, each metric can no longer discriminate among the remaining shots by their likelihood of logical failure. The performance under
    full error detection is shown for $f=3$ surface code cultivation, with a discard rate of
    $\sim82.4\%$; no shots passed full error detection for the other protocols.}
    \label{fig:comparison}
\end{figure*}

Caliper uses this observation to traverse the hidden-syndrome space in Fig.~\ref{fig:caliper}.
Starting from $\hat\sigma_{h_0}$, it evaluates hidden syndromes that differ by one hidden detector value
and preferentially continues from those with low $w_1$. This directs the search toward the syndromes
that contribute most strongly to the estimated logical error probability. Caliper also permits
steps that increase $w_1$ during an initial exploration period, subject to certain constraints.
In the Supplementary Material, we extend the search to exhaust the frontier and increase its
budget, finding only marginal reductions in logical error rate in the regimes of interest at a
substantial increase in evaluation time. 

Despite exploring a large number of hidden
syndromes, for $f=3$ surface code cultivation, Caliper spends only tens of nanoseconds on each prediction, thanks to the precomputed
graph distances. The resulting latency is within a factor of two of computing the complementary
gap in closed-boundary postselection (End Matter).

\textit{Numerical simulations.---} With Stim \cite{gidney_stim_2021}, we simulate $\ket{\bar Y}$ state cultivation---an efficient proxy
for $\ket{\bar T}$ state simulation---using the circuits of Refs.~\cite{gidney_magic_2024,
sahay2026fold, scc_fd5_simulation}, which we respectively refer to as color and (fold-transversal) surface code cultivation. We use the
same circuits and circuit-level depolarizing noise model (SD6), with physical error rate
$p=10^{-3}$, and compare the logical error rate achieved as a function of the escape-stage
postselection strength.

The results are shown in Fig.~\ref{fig:comparison}. Across all panels, Caliper achieves an equal or
lower logical error rate at any fixed discard rate than the other open-boundary methods tested. For the fault distance 3
($f=3$) color code cultivation, it performs nearly indistinguishably from the closed-boundary curve,
which most other methods also approximately follow. By contrast, in $f=3$ surface code cultivation,
the lowest achieved logical error rates are separated from closed-boundary performance by at least
an order of magnitude. Even full error detection on the visible syndromes achieves an error rate of $1.4\times10^{-5}$,
exceeding the target by almost an order of magnitude. Caliper alone smoothly
interpolates down to this performance. Our low-noise bound for full error detection (Supplementary
Material) suggests that the separation may reflect fundamental limits of the available information.

In the larger $f=5$ color code scheme, Caliper outperforms existing methods by orders of magnitude
at the relevant acceptance fractions. For instance, producing $\ket{\bar T}$ states at logical error rates of $10^{-7}$ is a common primitive in proposed fault-tolerant architectures \cite{gidney_magic_2024, zhou_resource_2025, sun2026foldedsurfacecodearchitecture, fazio2025lowoverheadmagicstatecircuits, khalid2026impactsdecoderlatencyutilityscale}.
Caliper reaches this target with an expected spacetime volume per accepted state only $10\%$
above that of closed-boundary postselection. The frozen gap \cite{birchall_macromux_2026},
the only other method we benchmarked to reach this error rate, requires approximately $7\times$
the closed-boundary volume. Reaching the more stringent $10^{-9}$ target with Caliper
requires an escape-stage acceptance fraction of $6.4\%$, compared with $23.7\%$ for
closed-boundary postselection. The expected spacetime volume is therefore approximately
$3.7\times$ larger. Even accepting only a few percent of shots, the $f=5$ surface code protocol \cite{fault_distance_4} under open-boundary postselection exhibits a logical error around $10^{-7}$.

A curious suggestion from Figure~\ref{fig:comparison} is that the performance gap between open- and closed-boundary postselection appears far larger for surface code cultivation than color code cultivation. As it turns out,
these protocols differ considerably in the size of their escape stages. For $f=3$ surface code
cultivation, the code is grown to $d=7$ and $r=3$ rounds of stabilizer measurements are performed
before deciding whether to keep the attempt. The $f=3$ color code escape stage instead uses $r=9$
rounds in a $d=15$ grafted code---roughly a $14\times$ larger syndrome. This motivates studying
modifications to the surface code escape stage.

Indeed, in Fig.~\ref{fig:modified-escape}, we independently vary the escape distance and number of QEC
cycles before running Caliper. Increasing the number of rounds alone leads to only an incremental
reduction in logical error rate, despite a $d=7$ surface code being large enough to maintain a rate
of order $10^{-6}$ per cycle. 
Recovering the target \(2\times10^{-6}\) requires increasing both the distance and the number of rounds, contrary to the usual expectation that \(d\) rounds of repeated syndrome extraction equate to a noiseless transversal measurement. Among the escape stages studied, the expected spacetime volume needed to reach this target is
minimized with $r=6$, $d=11$, and $50\%$ escape-stage acceptance. This results in an escape stage roughly
$5\times$ the size reported in Ref.~\cite{sahay2026fold}, translating to roughly a $5\times$
increase in expected spacetime volume for the entire protocol.

Finally, we study the discrepancy between open- and closed-boundary postselection, tracing it in part to rare shots in which the hidden syndrome substantially \textit{erodes} the closed-boundary gap (End Matter). These shots can be readily discarded using the readout but cannot be identified from visible information alone. We develop a phenomenological theory of this effect and hypothesize that a large escape stage is needed to keep the noisy interface between the cultivation and escape stages far from both the spatial boundaries of the code and the open temporal boundary in order to mitigate erosion.

\begin{figure}
    \centering
    \includegraphics[width=\linewidth]{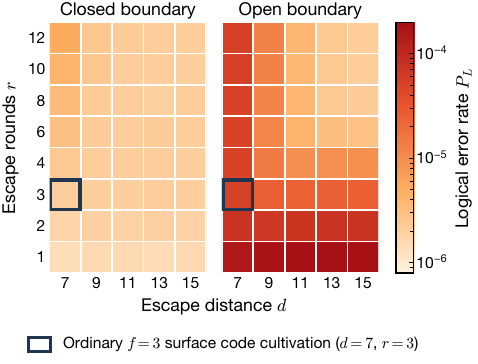}
    \caption{\textbf{Logical error rates under various escape stages.} The escape distance $d$ and
    number of escape rounds $r$ parameterize the volume of the escape stage. We evaluate the
    resulting LERs at $68\%$ escape-stage acceptance fraction, matching the rate reported in Ref.
    \cite{sahay2026fold} to reach $2 \times 10^{-6}$ for their $f=3$ surface code cultivation
    protocol. The outlined cells at $d=7, r=3$ correspond to the escape stage circuit proposed therein.}
    \label{fig:modified-escape}
\end{figure}

\textit{Outlook.---} Our results show that fast and accurate mid-circuit postselection can be
implemented on matchable codes. In some instances, Caliper recovers the logical error rates
reported under closed-boundary postselection at nearly unchanged spacetime cost. However, for
other protocols and target error rates, the cost of cultivation in practice may be significantly
larger than previously expected, raising the resource overheads of architectural proposals based on
its closed-boundary performance \cite{gidney2025factor, zhou_resource_2025, cain_shors_2026, webster2026pinnacle, xue2026factoring,
khan2026architecting, yoder2025tour}. Recent work has also developed exact simulations of non-Clifford $\ket{\bar T}$ cultivation
\cite{li2025soft, chase_clifft_2026, takada2026exactefficientsimulationnoisy,
hartweg_comparing_2026, wan_exact_2026, chan_diagnosing_2026, tuloup_computing_2026}, some finding discrepancies with $\ket{\bar Y}$
proxies. An exact simulation with open-boundary postselection is therefore an important direction
for determining how realistic escape stages affect these discrepancies.

Nevertheless, it may be possible to mitigate the cost of practical postselection by optimizing escape protocols. For example, adaptive escape schemes \cite{mishima2026adaptive} could reserve large escape stages for rare, high-risk shots (Supplementary Material), reducing the average overhead required to achieve high fidelities. Simultaneously, more sophisticated inference, including ensembling \cite{shutty_efficient_2026} and machine
learning \cite{gu_scalable_2026, bonilla_ataides_neural_2026, dentelski2026neural, blue_machine_2026,
baireuther_machine-learning-assisted_2017, baireuther_neural_2019, lacroix_scaling_2025,
bausch_learning_2024, meinerz_scalable_2022, bluvstein_fault-tolerant_2026,lange_data-driven_2025,
cao_generative_2025, liu_decoding_2025, hu2025efficient, zhang2025learning}, could improve on Caliper's performance.
Neural decoders, in particular, can provide postselection metrics based on decoder confidence that
approximate a logical gap and outperform the surface code complementary gap in closed-boundary
experiments \cite{gu_scalable_2026, dentelski2026neural, bausch_learning_2024}. Extending these metrics to open boundaries
poses challenges analogous to those for the complementary gap, and adapting them will require training
to estimate the probability of eventual decoding failure using only the visible syndrome. Caliper
can provide a benchmark for this task.

We expect the broader lessons on realistic postselection to extend to resource-state
preparation with high-rate quantum low-density parity-check codes \cite{menon_magic_2025, jacob_single-shot_2025, li_transversal_2025, hetenyi2026constantdepthmagicstate}. The analogue of an
escape stage for these protocols should likewise be co-designed with mid-circuit postselection.

\let\oldaddcontentsline\addcontentsline
\renewcommand{\addcontentsline}[3]{}
\begin{acknowledgments}
    \textit{Acknowledgements.---} We acknowledge helpful discussions with Andi Gu, Nazli Ugur
    Koyluoglu, Luke Coffman, Nishad Maskara, and Kaavya Sahay. We also thank the authors of Refs.
    \cite{gidney_magic_2024, sahay2026fold} for making their circuits and simulation data open
    source. The authors acknowledge the use of Anthropic's Claude Fable 5 to assist in writing the
    simulation code for this work. All code was reviewed, tested, and validated by the authors. 
    This work was supported by Anthropic under the AI for Science program, the National Science
    Foundation (CCF 2313084), the CUA Physics Frontiers Center (PHY-2317134) and NVQL (PHY-2410716),
    IARPA and the Army Research Office, under the Entangled Logical Qubits program (W911NF
    23-2-0219), the DARPA MeasQuIT program (HR0011-24-9-0359), and QuEra Computing.
\end{acknowledgments}

\textit{Data availability.---} The data supporting this work are available at
\url{https://github.com/rohanmehtagithub/MSC-data}.

\bibliographystyle{apsrev4-2}
\bibliography{references}

@misc{gidney2025factor,
      title={How to factor 2048 bit RSA integers with less than a million noisy qubits}, 
      author={Craig Gidney},
      year={2025},
      eprint={2505.15917},
      archivePrefix={arXiv},
      primaryClass={quant-ph},
      url={https://arxiv.org/abs/2505.15917}, 
}

@misc{birchall_macromux_2026,
	title = {Macromux: scalable postselection for high-threshold fault-tolerant quantum computation},
	shorttitle = {Macromux},
	url = {http://arxiv.org/abs/2603.04875},
	doi = {10.48550/arXiv.2603.04875},
	urldate = {2026-08-02},
	publisher = {arXiv},
	author = {Birchall, Patrick and Bridgeman, Jacob and Dawson, Christopher and Farrelly, Terry and Liu, Yehua and Nickerson, Naomi and Pant, Mihir and Roberts, Sam and Seetharam, Karthik and Tuckett, David},
	month = mar,
	year = {2026},
	note = {arXiv:2603.04875 [quant-ph]},
}

@article{higgott_sparse_2025,
	title = {Sparse {Blossom}: correcting a million errors per core second with minimum-weight matching},
	volume = {9},
	issn = {2521-327X},
	shorttitle = {Sparse {Blossom}},
	url = {http://arxiv.org/abs/2303.15933},
	doi = {10.22331/q-2025-01-20-1600},
	urldate = {2026-08-02},
	journal = {Quantum},
	author = {Higgott, Oscar and Gidney, Craig},
	month = jan,
	year = {2025},
	note = {arXiv:2303.15933 [quant-ph]},
	pages = {1600},
}

@article{dennis_topological_2002,
	title = {Topological quantum memory},
	volume = {43},
	issn = {0022-2488},
	url = {https://doi.org/10.1063/1.1499754},
	doi = {10.1063/1.1499754},
	number = {9},
	urldate = {2026-08-02},
	journal = {Journal of Mathematical Physics},
	author = {Dennis, Eric and Kitaev, Alexei and Landahl, Andrew and Preskill, John},
	month = sep,
	year = {2002},
	pages = {4452--4505},
}

@article{chamberland_very_2020,
	title = {Very low overhead fault-tolerant magic state preparation using redundant ancilla encoding and flag qubits},
	volume = {6},
	issn = {2056-6387},
	url = {http://arxiv.org/abs/2003.03049},
	doi = {10.1038/s41534-020-00319-5},
	number = {1},
	urldate = {2026-08-02},
	journal = {npj Quantum Information},
	author = {Chamberland, Christopher and Noh, Kyungjoo},
	month = oct,
	year = {2020},
	note = {arXiv:2003.03049 [quant-ph]},
	pages = {91},
}

@article{itogawa_efficient_2025,
	title = {Efficient {Magic} {State} {Distillation} by {Zero}-{Level} {Distillation}},
	volume = {6},
	issn = {2691-3399},
	url = {http://arxiv.org/abs/2403.03991},
	doi = {10.1103/thxx-njr6},
	number = {2},
	urldate = {2026-08-02},
	journal = {PRX Quantum},
	author = {Itogawa, Tomohiro and Takada, Yugo and Hirano, Yutaka and Fujii, Keisuke},
	month = jun,
	year = {2025},
	note = {arXiv:2403.03991 [quant-ph]},
	pages = {020356},
}

@misc{gidney_cleaner_2023,
	title = {Cleaner magic states with hook injection},
	url = {http://arxiv.org/abs/2302.12292},
	doi = {10.48550/arXiv.2302.12292},
	urldate = {2026-08-02},
	publisher = {arXiv},
	author = {Gidney, Craig},
	month = feb,
	year = {2023},
	note = {arXiv:2302.12292 [quant-ph]},
}

@article{li_magic_2015,
	title = {A magic state's fidelity can be superior to the operations that created it},
	volume = {17},
	issn = {1367-2630},
	url = {http://arxiv.org/abs/1410.7808},
	doi = {10.1088/1367-2630/17/2/023037},
	number = {2},
	urldate = {2026-08-02},
	journal = {New Journal of Physics},
	author = {Li, Ying},
	month = feb,
	year = {2015},
	note = {arXiv:1410.7808 [quant-ph]},
	pages = {023037},
}

@article{litinski_magic_2019,
	title = {Magic {State} {Distillation}: {Not} as {Costly} as {You} {Think}},
	volume = {3},
	issn = {2521-327X},
	shorttitle = {Magic {State} {Distillation}},
	url = {http://arxiv.org/abs/1905.06903},
	doi = {10.22331/q-2019-12-02-205},
	urldate = {2026-08-02},
	journal = {Quantum},
	author = {Litinski, Daniel},
	month = dec,
	year = {2019},
	note = {arXiv:1905.06903 [quant-ph]},
	pages = {205},
}

@article{zhou_methodology_2000,
	title = {Methodology for quantum logic gate constructions},
	volume = {62},
	issn = {1050-2947, 1094-1622},
	url = {http://arxiv.org/abs/quant-ph/0002039},
	doi = {10.1103/PhysRevA.62.052316},
	number = {5},
	urldate = {2026-08-02},
	journal = {Physical Review A},
	author = {Zhou, Xinlan and Leung, Debbie W. and Chuang, Isaac L.},
	month = oct,
	year = {2000},
	note = {arXiv:quant-ph/0002039},
	pages = {052316},
}

@article{hutter_efficient_2014,
	title = {Efficient {Markov} chain {Monte} {Carlo} algorithm for the surface code},
	volume = {89},
	url = {https://link.aps.org/doi/10.1103/PhysRevA.89.022326},
	doi = {10.1103/PhysRevA.89.022326},
	number = {2},
	urldate = {2026-08-02},
	journal = {Physical Review A},
	publisher = {American Physical Society},
	author = {Hutter, Adrian and Wootton, James R. and Loss, Daniel},
	month = feb,
	year = {2014},
	pages = {022326},
}

@article{bombin_fault-tolerant_2024,
	title = {Fault-{Tolerant} {Postselection} for {Low}-{Overhead} {Magic} {State} {Preparation}},
	volume = {5},
	url = {https://link.aps.org/doi/10.1103/PRXQuantum.5.010302},
	doi = {10.1103/PRXQuantum.5.010302},
	number = {1},
	urldate = {2026-08-02},
	journal = {PRX Quantum},
	publisher = {American Physical Society},
	author = {Bomb{\'i}n, H{\'e}ctor and Pant, Mihir and Roberts, Sam and Seetharam, Karthik I.},
	month = jan,
	year = {2024},
	pages = {010302},
}

@misc{gidney_yoked_2023,
	title = {Yoked surface codes},
	url = {http://arxiv.org/abs/2312.04522},
	doi = {10.48550/arXiv.2312.04522},
	urldate = {2026-08-02},
	publisher = {arXiv},
	author = {Gidney, Craig and Newman, Michael and Brooks, Peter and Jones, Cody},
	month = dec,
	year = {2023},
	note = {arXiv:2312.04522 [quant-ph]},
}

@article{english_thresholds_2025,
	title = {Thresholds for {Postselected} {Quantum} {Error} {Correction} from {Statistical} {Mechanics}},
	volume = {135},
	url = {https://link.aps.org/doi/10.1103/nh49-52y2},
	doi = {10.1103/nh49-52y2},
	number = {12},
	urldate = {2026-08-02},
	journal = {Physical Review Letters},
	publisher = {American Physical Society},
	author = {English, Lucas H. and Williamson, Dominic J. and Bartlett, Stephen D.},
	month = sep,
	year = {2025},
	pages = {120603},
}

@misc{meister_efficient_2024,
	title = {Efficient soft-output decoders for the surface code},
	url = {http://arxiv.org/abs/2405.07433},
	doi = {10.48550/arXiv.2405.07433},
	urldate = {2026-08-02},
	publisher = {arXiv},
	author = {Meister, Nadine and Pattison, Christopher A. and Preskill, John},
	month = jun,
	year = {2024},
	note = {arXiv:2405.07433 [quant-ph]},
}

@misc{wills_forced_2026,
	title = {Forced {Gap} {Post}-{Selection} for {Quantum} {LDPC} {Codes} and their {Operations}},
	url = {http://arxiv.org/abs/2605.20346},
	doi = {10.48550/arXiv.2605.20346},
	urldate = {2026-08-02},
	publisher = {arXiv},
	author = {Wills, Adam and Yoder, Theodore J. and Chuang, Isaac},
	month = jun,
	year = {2026},
	note = {arXiv:2605.20346 [quant-ph]},
}

@misc{xie_simple_2026,
	title = {Simple, {Efficient}, and {Generic} {Post}-{Selection} {Decoding} for {qLDPC} codes},
	url = {http://arxiv.org/abs/2601.17757},
	doi = {10.48550/arXiv.2601.17757},
	urldate = {2026-08-02},
	publisher = {arXiv},
	author = {Xie, Haipeng and Yoshioka, Nobuyuki and Tsubouchi, Kento and Li, Ying},
	month = jan,
	year = {2026},
	note = {arXiv:2601.17757 [quant-ph]},
}

@article{lee_efficient_2026,
	title = {Efficient {Post}-{Selection} for {General} {Quantum} {LDPC} {Codes}},
	volume = {12},
	issn = {2056-6387},
	url = {http://arxiv.org/abs/2510.05795},
	doi = {10.1038/s41534-026-01242-x},
	number = {1},
	urldate = {2026-08-02},
	journal = {npj Quantum Information},
	author = {Lee, Seok-Hyung and English, Lucas H. and Bartlett, Stephen D.},
	month = apr,
	year = {2026},
	note = {arXiv:2510.05795 [quant-ph]},
	pages = {96},
}

@misc{staples_scalable_2026,
	title = {Scalable {Postselection} of {Quantum} {Resources}},
	url = {http://arxiv.org/abs/2603.08697},
	doi = {10.48550/arXiv.2603.08697},
	urldate = {2026-08-02},
	publisher = {arXiv},
	author = {Staples, J. Wilson and Fu, Winston and Thompson, Jeff D.},
	month = mar,
	year = {2026},
	note = {arXiv:2603.08697 [quant-ph]},
}

@misc{chen_efficient_2025,
	title = {Efficient {Magic} {State} {Cultivation} on \${\textbackslash}mathbb\{{RP}\}{\textasciicircum}2\$},
	url = {http://arxiv.org/abs/2503.18657},
	doi = {10.48550/arXiv.2503.18657},
	urldate = {2026-08-02},
	publisher = {arXiv},
	author = {Chen, Zi-Han and Chen, Ming-Cheng and Lu, Chao-Yang and Pan, Jian-Wei},
	month = mar,
	year = {2025},
	note = {arXiv:2503.18657 [quant-ph]},
}

@misc{claes_cultivating_2025,
	title = {Cultivating {T} states on the surface code with only two-qubit gates},
	url = {http://arxiv.org/abs/2509.05232},
	doi = {10.48550/arXiv.2509.05232},
	urldate = {2026-08-02},
	publisher = {arXiv},
	author = {Claes, Jahan},
	month = sep,
	year = {2025},
	note = {arXiv:2509.05232 [quant-ph]},
}

@article{vaknin_high_2026,
	title = {High {Rate} {Magic} {State} {Cultivation} on the {Surface} {Code}},
	volume = {7},
	url = {https://link.aps.org/doi/10.1103/p8tw-6kq9},
	doi = {10.1103/p8tw-6kq9},
	number = {1},
	urldate = {2026-08-02},
	journal = {PRX Quantum},
	publisher = {American Physical Society},
	author = {Vaknin, Yotam and Jacoby, Shoham and Grimsmo, Arne and Retzker, Alex},
	month = mar,
	year = {2026},
	pages = {010353},
}

@misc{rosenfeld_magic_2025,
	title = {Magic state cultivation on a superconducting quantum processor},
	url = {http://arxiv.org/abs/2512.13908},
	doi = {10.48550/arXiv.2512.13908},
	urldate = {2026-08-02},
	publisher = {arXiv},
	author = {{Google Quantum AI} and Collaborators},
	month = dec,
	year = {2025},
	note = {arXiv:2512.13908 [quant-ph]},
}

@article{eastin_restrictions_2009,
	title = {Restrictions on {Transversal} {Encoded} {Quantum} {Gate} {Sets}},
	volume = {102},
	url = {https://link.aps.org/doi/10.1103/PhysRevLett.102.110502},
	doi = {10.1103/PhysRevLett.102.110502},
	number = {11},
	urldate = {2026-08-02},
	journal = {Physical Review Letters},
	publisher = {American Physical Society},
	author = {Eastin, Bryan and Knill, Emanuel},
	month = mar,
	year = {2009},
	pages = {110502},
}

@article{bravyi_universal_2005,
	title = {Universal quantum computation with ideal {Clifford} gates and noisy ancillas},
	volume = {71},
	copyright = {http://link.aps.org/licenses/aps-default-license},
	issn = {1050-2947, 1094-1622},
	url = {https://link.aps.org/doi/10.1103/PhysRevA.71.022316},
	doi = {10.1103/PhysRevA.71.022316},
	language = {en},
	number = {2},
	urldate = {2026-08-02},
	journal = {Physical Review A},
	author = {Bravyi, Sergey and Kitaev, Alexei},
	month = feb,
	year = {2005},
	pages = {022316},
}

@misc{gottesman_heisenberg_1998,
	title = {The {Heisenberg} {Representation} of {Quantum} {Computers}},
	url = {http://arxiv.org/abs/quant-ph/9807006},
	doi = {10.48550/arXiv.quant-ph/9807006},
	urldate = {2026-08-02},
	publisher = {arXiv},
	author = {Gottesman, Daniel},
	month = jul,
	year = {1998},
	note = {arXiv:quant-ph/9807006},
}

@misc{paetznick_demonstration_2024,
	title = {Demonstration of logical qubits and repeated error correction with better-than-physical error rates},
	url = {https://arxiv.org/abs/2404.02280v3},
	urldate = {2026-08-02},
	journal = {arXiv.org},
	author = {Paetznick, A. and others},
	month = apr,
	year = {2024},
}

@article{ryan-anderson_realization_2021,
	title = {Realization of {Real}-{Time} {Fault}-{Tolerant} {Quantum} {Error} {Correction}},
	volume = {11},
	url = {https://link.aps.org/doi/10.1103/PhysRevX.11.041058},
	doi = {10.1103/PhysRevX.11.041058},
	number = {4},
	urldate = {2026-08-02},
	journal = {Physical Review X},
	publisher = {American Physical Society},
	author = {Ryan-Anderson, C. and Bohnet, J. G. and Lee, K. and Gresh, D. and Hankin, A. and Gaebler, J. P. and Francois, D. and Chernoguzov, A. and Lucchetti, D. and Brown, N. C. and Gatterman, T. M. and Halit, S. K. and Gilmore, K. and Gerber, J. A. and Neyenhuis, B. and Hayes, D. and Stutz, R. P.},
	month = dec,
	year = {2021},
	pages = {041058},
}

@misc{reichardt_fault-tolerant_2025,
	title = {Fault-tolerant quantum computation with a neutral atom processor},
	url = {http://arxiv.org/abs/2411.11822},
	doi = {10.48550/arXiv.2411.11822},
	urldate = {2026-08-02},
	publisher = {arXiv},
	author = {Reichardt, Ben W. and others},
	month = jun,
	year = {2025},
	note = {arXiv:2411.11822 [quant-ph]},
}

@article{knill_resilient_1998,
	title = {Resilient quantum computation: error models and thresholds},
	volume = {454},
	issn = {1364-5021},
	shorttitle = {Resilient quantum computation},
	url = {https://doi.org/10.1098/rspa.1998.0166},
	doi = {10.1098/rspa.1998.0166},
	number = {1969},
	urldate = {2026-08-02},
	journal = {Proceedings of the Royal Society A: Mathematical, Physical and Engineering Sciences},
	author = {Knill, Emanuel and Laflamme, Raymond and Zurek, Wojciech H.},
	month = jan,
	year = {1998},
	pages = {365--384},
}

@article{abobeih_fault-tolerant_2022,
	title = {Fault-tolerant operation of a logical qubit in a diamond quantum processor},
	volume = {606},
	copyright = {2022 The Author(s)},
	issn = {1476-4687},
	url = {https://www.nature.com/articles/s41586-022-04819-6},
	doi = {10.1038/s41586-022-04819-6},
	language = {en},
	number = {7916},
	urldate = {2026-08-02},
	journal = {Nature},
	publisher = {Nature Publishing Group},
	author = {Abobeih, M. H. and Wang, Y. and Randall, J. and Loenen, S. J. H. and Bradley, C. E. and Markham, M. and Twitchen, D. J. and Terhal, B. M. and Taminiau, T. H.},
	month = jun,
	year = {2022},
	pages = {884--889},
}

@article{postler_demonstration_2022,
	title = {Demonstration of fault-tolerant universal quantum gate operations},
	volume = {605},
	copyright = {2022 The Author(s), under exclusive licence to Springer Nature Limited},
	issn = {1476-4687},
	url = {https://www.nature.com/articles/s41586-022-04721-1},
	doi = {10.1038/s41586-022-04721-1},
	language = {en},
	number = {7911},
	urldate = {2026-08-02},
	journal = {Nature},
	publisher = {Nature Publishing Group},
	author = {Postler, Lukas and Heu{\ss}en, Sascha and Pogorelov, Ivan and Rispler, Manuel and Feldker, Thomas and Meth, Michael and Marciniak, Christian D. and Stricker, Roman and Ringbauer, Martin and Blatt, Rainer and Schindler, Philipp and M{\"u}ller, Markus and Monz, Thomas},
	month = may,
	year = {2022},
	pages = {675--680},
}

@article{krinner_realizing_2022,
	title = {Realizing repeated quantum error correction in a distance-three surface code},
	volume = {605},
	copyright = {2022 The Author(s), under exclusive licence to Springer Nature Limited},
	issn = {1476-4687},
	url = {https://www.nature.com/articles/s41586-022-04566-8},
	doi = {10.1038/s41586-022-04566-8},
	language = {en},
	number = {7911},
	urldate = {2026-08-02},
	journal = {Nature},
	publisher = {Nature Publishing Group},
	author = {Krinner, Sebastian and Lacroix, Nathan and Remm, Ants and Di Paolo, Agustin and Genois, Elie and Leroux, Catherine and Hellings, Christoph and Lazar, Stefania and Swiadek, Francois and Herrmann, Johannes and Norris, Graham J. and Andersen, Christian Kraglund and M{\"u}ller, Markus and Blais, Alexandre and Eichler, Christopher and Wallraff, Andreas},
	month = may,
	year = {2022},
	pages = {669--674},
}

@article{acharya_quantum_2025,
	title = {Quantum error correction below the surface code threshold},
	volume = {638},
	copyright = {2024 The Author(s)},
	issn = {1476-4687},
	url = {https://www.nature.com/articles/s41586-024-08449-y},
	doi = {10.1038/s41586-024-08449-y},
	language = {en},
	number = {8052},
	urldate = {2026-08-02},
	journal = {Nature},
	publisher = {Nature Publishing Group},
	author = {{Google Quantum AI} and Collaborators},
	month = feb,
	year = {2025},
	pages = {920--926},
}

@article{acharya_suppressing_2023,
	title = {Suppressing quantum errors by scaling a surface code logical qubit},
	volume = {614},
	copyright = {2023 The Author(s)},
	issn = {1476-4687},
	url = {https://www.nature.com/articles/s41586-022-05434-1},
	doi = {10.1038/s41586-022-05434-1},
	language = {en},
	number = {7949},
	urldate = {2026-08-02},
	journal = {Nature},
	publisher = {Nature Publishing Group},
	author = {{Google Quantum AI}},
	month = feb,
	year = {2023},
	pages = {676--681},
}

@article{bluvstein_fault-tolerant_2026,
	title = {A fault-tolerant neutral-atom architecture for universal quantum computation},
	volume = {649},
	copyright = {2025 The Author(s)},
	issn = {1476-4687},
	url = {https://www.nature.com/articles/s41586-025-09848-5},
	doi = {10.1038/s41586-025-09848-5},
	language = {en},
	number = {8095},
	urldate = {2026-08-02},
	journal = {Nature},
	publisher = {Nature Publishing Group},
	author = {Bluvstein, Dolev and Geim, Alexandra A. and Li, Sophie H. and Evered, Simon J. and Bonilla Ataides, J. Pablo and Baranes, Gefen and Gu, Andi and Manovitz, Tom and Xu, Muqing and Kalinowski, Marcin and Majidy, Shayan and Kokail, Christian and Maskara, Nishad and Trapp, Elias C. and Stewart, Luke M. and Hollerith, Simon and Zhou, Hengyun and Gullans, Michael J. and Yelin, Susanne F. and Greiner, Markus and Vuleti{\'c}, Vladan and Cain, Madelyn and Lukin, Mikhail D.},
	month = jan,
	year = {2026},
	pages = {39--46},
}

@article{terhal_quantum_2015,
	title = {Quantum error correction for quantum memories},
	volume = {87},
	url = {https://link.aps.org/doi/10.1103/RevModPhys.87.307},
	doi = {10.1103/RevModPhys.87.307},
	number = {2},
	urldate = {2026-08-02},
	journal = {Reviews of Modern Physics},
	publisher = {American Physical Society},
	author = {Terhal, Barbara M.},
	month = apr,
	year = {2015},
	pages = {307--346},
}

@article{aharonov_fault-tolerant_2008,
	title = {Fault-{Tolerant} {Quantum} {Computation} with {Constant} {Error} {Rate}},
	volume = {38},
	issn = {0097-5397},
	url = {https://epubs.siam.org/doi/10.1137/S0097539799359385},
	doi = {10.1137/S0097539799359385},
	number = {4},
	urldate = {2026-08-02},
	journal = {SIAM Journal on Computing},
	publisher = {Society for Industrial and Applied Mathematics},
	author = {Aharonov, Dorit and Ben-Or, Michael},
	month = jan,
	year = {2008},
	pages = {1207--1282},
}

@article{steane_error_1996,
	title = {Error {Correcting} {Codes} in {Quantum} {Theory}},
	volume = {77},
	url = {https://link.aps.org/doi/10.1103/PhysRevLett.77.793},
	doi = {10.1103/PhysRevLett.77.793},
	number = {5},
	urldate = {2026-08-02},
	journal = {Physical Review Letters},
	publisher = {American Physical Society},
	author = {Steane, A. M.},
	month = jul,
	year = {1996},
	pages = {793--797},
}

@article{shor_scheme_1995,
	title = {Scheme for reducing decoherence in quantum computer memory},
	volume = {52},
	url = {https://link.aps.org/doi/10.1103/PhysRevA.52.R2493},
	doi = {10.1103/PhysRevA.52.R2493},
	number = {4},
	urldate = {2026-08-02},
	journal = {Physical Review A},
	publisher = {American Physical Society},
	author = {Shor, Peter W.},
	month = oct,
	year = {1995},
	pages = {R2493--R2496},
}

@article{sahay2026fold,
  title = {Fold-transversal surface code cultivation},
  author = {Sahay, Kaavya and Tsai, Pei-Kai and Chang, Kathleen (Katie) and Su, Qile and Smith, Thomas B. and Singh, Shraddha and Puri, Shruti},
  journal = {PRX Quantum},
  volume = {7},
  issue = {3},
  pages = {033006},
  numpages = {26},
  year = {2026},
  month = {Jul},
  publisher = {American Physical Society},
  doi = {10.1103/gpvl-lg4c},
  url = {https://link.aps.org/doi/10.1103/gpvl-lg4c}
}

@misc{pesah_fault-tolerant_2025,
	title = {Fault-tolerant transformations of spacetime codes},
	url = {http://arxiv.org/abs/2509.09603},
	doi = {10.48550/arXiv.2509.09603},
	urldate = {2026-04-15},
	publisher = {arXiv},
	author = {Pesah, Arthur and Daniel, Austin K. and Tzitrin, Ilan and Vasmer, Michael},
	month = sep,
	year = {2025},
	note = {arXiv:2509.09603 [quant-ph]},
}

@misc{cain_shors_2026,
	title = {Shor's algorithm is possible with as few as 10,000 reconfigurable atomic qubits},
	url = {http://arxiv.org/abs/2603.28627},
	doi = {10.48550/arXiv.2603.28627},
	urldate = {2026-04-14},
	publisher = {arXiv},
	author = {Cain, Madelyn and Xu, Qian and King, Robbie and Picard, Lewis R. B. and Levine, Harry and Endres, Manuel and Preskill, John and Huang, Hsin-Yuan and Bluvstein, Dolev},
	month = mar,
	year = {2026},
	note = {arXiv:2603.28627 [quant-ph]},
}

@article{gottesman_demonstrating_1999,
	title = {Demonstrating the viability of universal quantum computation using teleportation and single-qubit operations},
	volume = {402},
	copyright = {1999 Macmillan Magazines Ltd.},
	issn = {1476-4687},
	url = {https://www.nature.com/articles/46503},
	doi = {10.1038/46503},
	language = {en},
	number = {6760},
	urldate = {2026-01-28},
	journal = {Nature},
	publisher = {Nature Publishing Group},
	author = {Gottesman, Daniel and Chuang, Isaac L.},
	month = nov,
	year = {1999},
	pages = {390--393},
}

@article{gidney_stim_2021,
	title = {Stim: a fast stabilizer circuit simulator},
	volume = {5},
	issn = {2521-327X},
	shorttitle = {Stim},
	url = {http://arxiv.org/abs/2103.02202},
	doi = {10.22331/q-2021-07-06-497},
	urldate = {2026-01-26},
	journal = {Quantum},
	author = {Gidney, Craig},
	month = jul,
	year = {2021},
	note = {arXiv:2103.02202 [quant-ph]},
	pages = {497},
}

@misc{gu_scalable_2026,
      title={Scalable Neural Decoders for Practical Fault-Tolerant Quantum Computation}, 
      author={Andi Gu and J. Pablo Bonilla Ataides and Mikhail D. Lukin and Susanne F. Yelin},
      year={2026},
      eprint={2604.08358},
      archivePrefix={arXiv},
      primaryClass={quant-ph},
      url={https://arxiv.org/abs/2604.08358}, 
}

@misc{jacob_single-shot_2025,
	title = {Single-{Shot} {Decoding} and {Fault}-tolerant {Gates} with {Trivariate} {Tricycle} {Codes}},
	url = {http://arxiv.org/abs/2508.08191},
	doi = {10.48550/arXiv.2508.08191},
	urldate = {2026-01-14},
	publisher = {arXiv},
	author = {Jacob, Abraham and McLauchlan, Campbell and Browne, Dan E.},
	month = aug,
	year = {2025},
	note = {arXiv:2508.08191 [quant-ph]},
}

@misc{li_transversal_2025,
	title = {Transversal dimension jump for product {qLDPC} codes},
	url = {http://arxiv.org/abs/2510.07269},
	doi = {10.48550/arXiv.2510.07269},
	urldate = {2026-01-14},
	publisher = {arXiv},
	author = {Li, Christine and Preskill, John and Xu, Qian},
	month = oct,
	year = {2025},
	note = {arXiv:2510.07269 [quant-ph]},
}

@article{menon_magic_2025,
	title = {Magic {Tricycles}: {Efficient} {Magic}-{State} {Generation} with {Finite} {Block}-{Length} {Quantum} {LDPC} {Codes}},
	shorttitle = {Magic tricycles},
	author = {Menon, Varun and Bonilla Ataides, J. Pablo and Mehta, Rohan and Gu, Andi and Tan, Daniel Bochen and Lukin, Mikhail D.},
	journal = {Physical Review X},
	volume = {16},
	number = {2},
	pages = {021014},
	year = {2026},
	month = apr,
	doi = {10.1103/ghhp-cytl},
	url = {https://doi.org/10.1103/ghhp-cytl},
	eprint = {2508.10714},
	archivePrefix = {arXiv},
	note = {arXiv:2508.10714},
}

@misc{gidney_magic_2024,
	title = {Magic state cultivation: growing {T} states as cheap as {CNOT} gates},
	shorttitle = {Magic state cultivation},
	url = {http://arxiv.org/abs/2409.17595},
	doi = {10.48550/arXiv.2409.17595},
	urldate = {2026-01-14},
	publisher = {arXiv},
	author = {Gidney, Craig and Shutty, Noah and Jones, Cody},
	month = sep,
	year = {2024},
	note = {arXiv:2409.17595 [quant-ph]},
}

@article{kitaev_fault-tolerant_2003,
	title = {Fault-tolerant quantum computation by anyons},
	volume = {303},
	issn = {00034916},
	url = {http://arxiv.org/abs/quant-ph/9707021},
	doi = {10.1016/S0003-4916(02)00018-0},
	number = {1},
	urldate = {2026-01-14},
	journal = {Annals of Physics},
	author = {Kitaev, A. Yu},
	month = jan,
	year = {2003},
	note = {arXiv:quant-ph/9707021},
	pages = {2--30},
}

@article{fowler_surface_2012,
	title = {Surface codes: {Towards} practical large-scale quantum computation},
	volume = {86},
	copyright = {http://link.aps.org/licenses/aps-default-license},
	issn = {1050-2947, 1094-1622},
	shorttitle = {Surface codes},
	url = {https://link.aps.org/doi/10.1103/PhysRevA.86.032324},
	doi = {10.1103/PhysRevA.86.032324},
	language = {en},
	number = {3},
	urldate = {2026-01-14},
	journal = {Physical Review A},
	author = {Fowler, Austin G. and Mariantoni, Matteo and Martinis, John M. and Cleland, Andrew N.},
	month = sep,
	year = {2012},
	pages = {032324},
}

@article{bluvstein_logical_2024,
	title = {Logical quantum processor based on reconfigurable atom arrays},
	volume = {626},
	copyright = {2023 The Author(s)},
	issn = {1476-4687},
	url = {https://www.nature.com/articles/s41586-023-06927-3},
	doi = {10.1038/s41586-023-06927-3},
	language = {en},
	number = {7997},
	urldate = {2026-01-14},
	journal = {Nature},
	publisher = {Nature Publishing Group},
	author = {Bluvstein, Dolev and Evered, Simon J. and Geim, Alexandra A. and Li, Sophie H. and Zhou, Hengyun and Manovitz, Tom and Ebadi, Sepehr and Cain, Madelyn and Kalinowski, Marcin and Hangleiter, Dominik and Bonilla Ataides, J. Pablo and Maskara, Nishad and Cong, Iris and Gao, Xun and Sales Rodriguez, Pedro and Karolyshyn, Thomas and Semeghini, Giulia and Gullans, Michael J. and Greiner, Markus and Vuleti{\'c}, Vladan and Lukin, Mikhail D.},
	month = feb,
	year = {2024},
	pages = {58--65},
}

@misc{mishima2026adaptive,
      title={Adaptive Window Decoding based on Spatiotemporal Complementary Gap}, 
      author={Moeto Mishima and Riki Toshio and Kaito Kishi and Jun Fujisaki and Hirotaka Oshima and Shintaro Sato and Keisuke Fujii},
      year={2026},
      eprint={2605.14637},
      archivePrefix={arXiv},
      primaryClass={quant-ph},
      url={https://arxiv.org/abs/2605.14637}, 
}

@misc{hirano2026efficient,
      title={Efficient magic state cultivation with lattice surgery}, 
      author={Yutaka Hirano and Riki Toshio and Tomohiro Itogawa and Keisuke Fujii},
      year={2026},
      eprint={2510.24615},
      archivePrefix={arXiv},
      primaryClass={quant-ph},
      url={https://arxiv.org/abs/2510.24615}, 
}

@misc{chen2026scalable,
      title={Scalable accuracy gains from postselection in quantum error correcting codes}, 
      author={Hongkun Chen and Daohong Xu and Grace M. Sommers and David A. Huse and Jeff D. Thompson and Sarang Gopalakrishnan},
      year={2026},
      eprint={2510.05222},
      archivePrefix={arXiv},
      primaryClass={cond-mat.stat-mech},
      url={https://arxiv.org/abs/2510.05222}, 
}

@article{iyer_hardness_2015,
    title = {Hardness of {Decoding} {Quantum} {Stabilizer} {Codes}},
    volume = {61},
    issn = {1557-9654},
    url = {https://ieeexplore.ieee.org/document/7097029},
    doi = {10.1109/TIT.2015.2422294},
    number = {9},
    urldate = {2026-09-14},
    journal = {IEEE Transactions on Information Theory},
    author = {Iyer, Pavithran and Poulin, David},
    month = sep,
    year = {2015},
    pages = {5209--5223},
}

@article{shutty_efficient_2026,
    title = {Efficient {Near}-{Optimal} {Decoding} of the {Surface} {Code} through {Ensembling}},
    volume = {136},
    url = {https://link.aps.org/doi/10.1103/77j6-32xx},
    doi = {10.1103/77j6-32xx},
    number = {7},
    urldate = {2026-09-14},
    journal = {Physical Review Letters},
    publisher = {American Physical Society},
    author = {Shutty, Noah and Newman, Michael and Villalonga, Benjamin},
    month = feb,
    year = {2026},
    pages = {070603},
}

@article{bonilla_ataides_neural_2026,
    title = {Neural {Decoders} for {Universal} {Quantum} {Algorithms}},
    volume = {1},
    url = {https://link.aps.org/doi/10.1103/vjn1-mbxl},
    doi = {10.1103/vjn1-mbxl},
    number = {1},
    urldate = {2026-09-14},
    journal = {PRX Intelligence},
    publisher = {American Physical Society},
    author = {Bonilla Ataides, J. Pablo and Gu, Andi and Yelin, Susanne F. and Lukin, Mikhail D.},
    month = aug,
    year = {2026},
    pages = {013008},
}

@article{blue_machine_2026,
    title = {Machine {Learning} {Decoding} of {Circuit}-{Level} {Noise} for {Bivariate} {Bicycle} {Codes}},
    volume = {10},
    url = {https://quantum-journal.org/papers/q-2026-06-30-2149/},
    doi = {10.22331/q-2026-06-30-2149},
    language = {en-GB},
    urldate = {2026-09-14},
    journal = {Quantum},
    publisher = {Verein zur Förderung des Open Access Publizierens in den Quantenwissenschaften},
    author = {Blue, John and Avlani, Harshil and He, Zhiyang and Ziyin, Liu and Chuang, Isaac L.},
    month = jun,
    year = {2026},
    pages = {2149},
}

@misc{baireuther_machine-learning-assisted_2017,
    title = {Machine-learning-assisted correction of correlated qubit errors in a topological code},
    url = {https://arxiv.org/abs/1705.07855v3},
    doi = {10.22331/q-2018-01-29-48},
    language = {en},
    urldate = {2026-09-14},
    journal = {arXiv.org},
    author = {Baireuther, P. and O'Brien, T. E. and Tarasinski, B. and Beenakker, C. W. J.},
    month = may,
    year = {2017},
}

@article{baireuther_neural_2019,
    title = {Neural network decoder for topological color codes with circuit level noise},
    volume = {21},
    issn = {1367-2630},
    url = {https://doi.org/10.1088/1367-2630/aaf29e},
    doi = {10.1088/1367-2630/aaf29e},
    language = {en},
    number = {1},
    urldate = {2026-09-14},
    journal = {New Journal of Physics},
    publisher = {IOP Publishing},
    author = {Baireuther, P and Caio, M D and Criger, B and Beenakker, C W J and O’Brien, T E},
    month = jan,
    year = {2019},
    pages = {013003},
}

@article{lacroix_scaling_2025,
    title = {Scaling and logic in the colour code on a superconducting quantum processor},
    volume = {645},
    copyright = {2025 The Author(s)},
    issn = {1476-4687},
    url = {https://www.nature.com/articles/s41586-025-09061-4},
    doi = {10.1038/s41586-025-09061-4},
    language = {en},
    number = {8081},
    urldate = {2026-09-14},
    journal = {Nature},
    publisher = {Nature Publishing Group},
    author = {Lacroix, N. and others},
    month = sep,
    year = {2025},
    pages = {614--619},
}

@article{bausch_learning_2024,
    title = {Learning high-accuracy error decoding for quantum processors},
    volume = {635},
    copyright = {2024 The Author(s)},
    issn = {1476-4687},
    url = {https://www.nature.com/articles/s41586-024-08148-8},
    doi = {10.1038/s41586-024-08148-8},
    language = {en},
    number = {8040},
    urldate = {2026-09-14},
    journal = {Nature},
    publisher = {Nature Publishing Group},
    author = {Bausch, Johannes and Senior, Andrew W. and Heras, Francisco J. H. and Edlich, Thomas and Davies, Alex and Newman, Michael and Jones, Cody and Satzinger, Kevin and Niu, Murphy Yuezhen and Blackwell, Sam and Holland, George and Kafri, Dvir and Atalaya, Juan and Gidney, Craig and Hassabis, Demis and Boixo, Sergio and Neven, Hartmut and Kohli, Pushmeet},
    month = nov,
    year = {2024},
    pages = {834--840},
}

@article{meinerz_scalable_2022,
    title = {Scalable {Neural} {Decoder} for {Topological} {Surface} {Codes}},
    volume = {128},
    url = {https://link.aps.org/doi/10.1103/PhysRevLett.128.080505},
    doi = {10.1103/PhysRevLett.128.080505},
    number = {8},
    urldate = {2026-09-14},
    journal = {Physical Review Letters},
    publisher = {American Physical Society},
    author = {Meinerz, Kai and Park, Chae-Yeun and Trebst, Simon},
    month = feb,
    year = {2022},
    pages = {080505},
}

@article{lange_data-driven_2025,
    title = {Data-driven decoding of quantum error correcting codes using graph neural networks},
    volume = {7},
    url = {https://link.aps.org/doi/10.1103/PhysRevResearch.7.023181},
    doi = {10.1103/PhysRevResearch.7.023181},
    number = {2},
    urldate = {2026-09-14},
    journal = {Physical Review Research},
    publisher = {American Physical Society},
    author = {Lange, Moritz and Havström, Pontus and Srivastava, Basudha and Bengtsson, Isak and Bergentall, Valdemar and Hammar, Karl and Heuts, Olivia and van Nieuwenburg, Evert and Granath, Mats},
    month = may,
    year = {2025},
    pages = {023181},
}

@misc{cao_generative_2025,
    title = {Generative {Decoding} for {Quantum} {Error}-correcting {Codes}},
    url = {http://arxiv.org/abs/2503.21374},
    doi = {10.48550/arXiv.2503.21374},
    urldate = {2026-09-14},
    publisher = {arXiv},
    author = {Cao, Hanyan and Pan, Feng and Feng, Dongyang and Wang, Yijia and Zhang, Pan},
    month = mar,
    year = {2025},
    note = {arXiv:2503.21374 [quant-ph]},
}

@misc{liu_decoding_2025,
    title = {Decoding quantum low density parity check codes with diffusion},
    url = {https://arxiv.org/abs/2509.22347v1},
    language = {en},
    urldate = {2026-09-14},
    journal = {arXiv.org},
    author = {Liu, Zejun and Gong, Anqi and Clark, Bryan K.},
    month = sep,
    year = {2025},
}

@misc{hu2025efficient,
      title={Efficient and Universal Neural-Network Decoder for Stabilizer-Based Quantum Error Correction}, 
      author={Gengyuan Hu and Wanli Ouyang and Chao-Yang Lu and Chen Lin and Han-Sen Zhong},
      year={2025},
      eprint={2502.19971},
      archivePrefix={arXiv},
      primaryClass={quant-ph},
      url={https://arxiv.org/abs/2502.19971}, 
}

@misc{zhang2025learning,
      title={Learning Neural Decoding with Parallelism and Self-Coordination for Quantum Error Correction}, 
      author={Kai Zhang and Situ Wang and Linghang Kong and Fang Zhang and Zhengfeng Ji and Jianxin Chen},
      year={2025},
      eprint={2509.03815},
      archivePrefix={arXiv},
      primaryClass={quant-ph},
      url={https://arxiv.org/abs/2509.03815}, 
}

@misc{dentelski2026neural,
      title={Neural network decoder confidence as a learned proxy for the logical gap}, 
      author={David Dentelski},
      year={2026},
      eprint={2606.08758},
      archivePrefix={arXiv},
      primaryClass={quant-ph},
      url={https://arxiv.org/abs/2606.08758}, 
}

@misc{computing2026quantumerrorcorrectiontoric,
      title={Quantum error correction with the toric code}, 
      author={{Atom Computing} and Collaborators},
      year={2026},
      eprint={2606.04079},
      archivePrefix={arXiv},
      primaryClass={quant-ph},
      url={https://arxiv.org/abs/2606.04079}, 
}

@misc{hetenyi2026constantdepthmagicstate,
      title={Constant depth magic state cultivation with Clifford measurements by gauging}, 
      author={Bence Hetényi and Benjamin J. Brown and Dominic J. Williamson},
      year={2026},
      eprint={2603.05429},
      archivePrefix={arXiv},
      primaryClass={quant-ph},
      url={https://arxiv.org/abs/2603.05429}, 
}

@article{webster2026pinnacle,
  title={The Pinnacle Architecture: Reducing the cost of breaking RSA-2048 to 100 000 physical qubits using quantum LDPC codes},
  author={Webster, Paul and Berent, Lucas and Chandra, Omprakash and Hockings, Evan T and Baspin, Nou{\'e}dyn and Thomsen, Felix and Smith, Samuel C and Cohen, Lawrence Z},
  journal={arXiv preprint arXiv:2602.11457},
  year={2026}
}

@article{xue2026factoring,
  title={Factoring $2048 $ bit RSA integers with a half-million-qubit modular atomic processor},
  author={Xue, Tian and Covey, Jacob P},
  journal={arXiv preprint arXiv:2605.03951},
  year={2026}
}

@article{khan2026architecting,
  title={Architecting Early Fault Tolerant Neutral Atoms Systems with Quantum Advantage},
  author={Khan, Sahil and Sethi, Sayam and Sahay, Kaavya and Lin, Yingjia and Alnas, Jude and Kurapati, Suhas and Anand, Abhinav and Baker, Jonathan M and Brown, Kenneth R},
  journal={arXiv preprint arXiv:2604.19735},
  year={2026}
}

@article{yoder2025tour,
  title={Tour de gross: A modular quantum computer based on bivariate bicycle codes},
  author={Yoder, Theodore J and Schoute, Eddie and Rall, Patrick and Pritchett, Emily and Gambetta, Jay M and Cross, Andrew W and Carroll, Malcolm and Beverland, Michael E},
  journal={arXiv preprint arXiv:2506.03094},
  year={2025}
}

@article{li2025soft,
  title={SOFT: a high-performance simulator for universal fault-tolerant quantum circuits},
  author={Li, Riling and Zheng, Keli and Zhang, Yiming and Lou, Huazhe and Ying, Shenggang and Liu, Ke and Sun, Xiaoming},
  journal={arXiv preprint arXiv:2512.23037},
  year={2025}
}

@misc{takada2026exactefficientsimulationnoisy,
      title={Exact efficient simulation of noisy logical magic states using Clifford stabilizers}, 
      author={Yugo Takada and Stephen D. Bartlett and Dominic J. Williamson},
      year={2026},
      eprint={2609.16929},
      archivePrefix={arXiv},
      primaryClass={quant-ph},
      url={https://arxiv.org/abs/2609.16929}, 
}

@misc{hartweg_comparing_2026,
    title = {Comparing magic state cultivation methods using matrix product states},
    url = {http://arxiv.org/abs/2609.19116},
    doi = {10.48550/arXiv.2609.19116},
    urldate = {2026-09-17},
    publisher = {arXiv},
    author = {Hartweg, Tom and Orioli, Asier Piñeiro},
    month = sep,
    year = {2026},
    note = {arXiv:2609.19116 [quant-ph]},
}

@misc{wan_exact_2026,
    title = {Exact logical error rates for magic state cultivation},
    url = {http://arxiv.org/abs/2609.18922},
    doi = {10.48550/arXiv.2609.18922},
    urldate = {2026-09-17},
    publisher = {arXiv},
    author = {Wan, Kwok Ho and Zapirain, Ainhoa},
    month = sep,
    year = {2026},
    note = {arXiv:2609.18922 [quant-ph]},
}

@misc{chase_clifft_2026,
    title = {Clifft: {Fast} {Exact} {Simulation} of {Near}-{Clifford} {Quantum} {Circuits}},
    shorttitle = {Clifft},
    url = {http://arxiv.org/abs/2604.27058},
    doi = {10.48550/arXiv.2604.27058},
    urldate = {2026-09-17},
    publisher = {arXiv},
    author = {Chase, Bradley A. and Labib, Farrokh},
    month = may,
    year = {2026},
    note = {arXiv:2604.27058 [quant-ph]},
}

@misc{chan_diagnosing_2026,
    title = {Diagnosing and {Restoring} the {Degraded} {Fault} {Distance} of {Magic} {State} {Cultivation}},
    url = {http://arxiv.org/abs/2609.17706},
    doi = {10.48550/arXiv.2609.17706},
    urldate = {2026-09-17},
    publisher = {arXiv},
    author = {Chan, Tim and Strikis, Armands and Sun, Zhu and Cai, Zhenyu},
    month = sep,
    year = {2026},
    note = {arXiv:2609.17706 [quant-ph]},
}

@article{higgott_pymatching_2022,
	title = {{PyMatching: A Python package for decoding quantum codes with minimum-weight perfect matching}},
	author = {Oscar Higgott},
	journal = {ACM Transactions on Quantum Computing},
	volume = {3},
	number = {3},
	pages = {1--16},
	year = {2022},
	month = {sep},
	doi = {10.1145/3505637},
	url = {https://doi.org/10.1145/3505637},
	eprint = {2105.13082},
	archivePrefix = {arXiv},
	note = {arXiv:2105.13082},
}

@misc{gidney_flexible_2019,
	title = {{Flexible layout of surface code computations using AutoCCZ states}},
	author = {Craig Gidney and Austin G. Fowler},
	year = {2019},
	publisher = {arXiv},
	url = {https://arxiv.org/abs/1905.08916},
	doi = {10.48550/arXiv.1905.08916},
	eprint = {1905.08916},
	archivePrefix = {arXiv},
	note = {arXiv:1905.08916 [quant-ph]},
}

@inproceedings{wu_micro_2025,
	title = {{Micro Blossom: Accelerated Minimum-Weight Perfect Matching Decoding for Quantum Error Correction}},
	author = {Yue Wu and Namitha Liyanage and Lin Zhong},
	booktitle = {Proceedings of the 30th ACM International Conference on Architectural Support for Programming Languages and Operating Systems, Volume 2},
	pages = {639--654},
	year = {2025},
	month = {mar},
	doi = {10.1145/3676641.3716005},
	url = {https://doi.org/10.1145/3676641.3716005},
	eprint = {2502.14787},
	archivePrefix = {arXiv},
	note = {arXiv:2502.14787},
}

@article{smith_mitigating_2024,
	title = {{Mitigating errors in logical qubits}},
	author = {Samuel C. Smith and Benjamin J. Brown and Stephen D. Bartlett},
	journal = {Communications Physics},
	volume = {7},
	number = {1},
	pages = {386},
	year = {2024},
	month = {nov},
	doi = {10.1038/s42005-024-01883-4},
	url = {https://doi.org/10.1038/s42005-024-01883-4},
	eprint = {2405.03766},
	archivePrefix = {arXiv},
	note = {arXiv:2405.03766},
}

@misc{fowler_time-optimal_2012,
	title = {{Time-optimal quantum computation}},
	author = {Austin G. Fowler},
	year = {2013},
	publisher = {arXiv},
	url = {https://arxiv.org/abs/1210.4626},
	doi = {10.48550/arXiv.1210.4626},
	eprint = {1210.4626},
	archivePrefix = {arXiv},
	note = {arXiv:1210.4626 [quant-ph]},
}

@article{duclos-cianci_fast_2010,
	title = {{Fast Decoders for Topological Quantum Codes}},
	author = {Guillaume Duclos-Cianci and David Poulin},
	journal = {Physical Review Letters},
	volume = {104},
	number = {5},
	pages = {050504},
	year = {2010},
	month = {feb},
	doi = {10.1103/PhysRevLett.104.050504},
	url = {https://doi.org/10.1103/PhysRevLett.104.050504},
	eprint = {0911.0581},
	archivePrefix = {arXiv},
	note = {arXiv:0911.0581},
}

@inproceedings{das_scalable_2022,
	title = {{AFS: Accurate, Fast, and Scalable Error-Decoding for Fault-Tolerant Quantum Computers}},
	author = {Poulami Das and Christopher A. Pattison and Srilatha Manne and Douglas Carmean and Krysta Svore and Moinuddin Qureshi and Nicolas Delfosse},
	booktitle = {2022 IEEE International Symposium on High-Performance Computer Architecture (HPCA)},
	pages = {259--273},
	year = {2022},
	month = {apr},
	doi = {10.1109/hpca53966.2022.00027},
	url = {https://doi.org/10.1109/hpca53966.2022.00027},
	eprint = {2001.06598},
	archivePrefix = {arXiv},
	note = {arXiv:2001.06598},
}

@inproceedings{holmes_nisq_2020,
	title = {{NISQ+: Boosting quantum computing power by approximating quantum error correction}},
	author = {Adam Holmes and Mohammad Reza Jokar and Ghasem Pasandi and Yongshan Ding and Massoud Pedram and Frederic T. Chong},
	booktitle = {2020 ACM/IEEE 47th Annual International Symposium on Computer Architecture (ISCA)},
	pages = {556--569},
	year = {2020},
	month = {may},
	doi = {10.1109/isca45697.2020.00053},
	url = {https://doi.org/10.1109/isca45697.2020.00053},
	eprint = {2004.04794},
	archivePrefix = {arXiv},
	note = {arXiv:2004.04794},
}

@inproceedings{liyanage_scalable_2023,
	title = {{Scalable Quantum Error Correction for Surface Codes using FPGA}},
	author = {Namitha Liyanage and Yue Wu and Alexander Deters and Lin Zhong},
	booktitle = {2023 IEEE International Conference on Quantum Computing and Engineering (QCE)},
	pages = {916--927},
	year = {2023},
	month = {sep},
	doi = {10.1109/qce57702.2023.00106},
	url = {https://doi.org/10.1109/qce57702.2023.00106},
	eprint = {2301.08419},
	archivePrefix = {arXiv},
	note = {arXiv:2301.08419},
}

@article{nishimori_internal_1981,
	title = {{Internal Energy, Specific Heat and Correlation Function of the Bond-Random Ising Model}},
	author = {Nishimori, Hidetoshi},
	journal = {Progress of Theoretical Physics},
	volume = {66},
	number = {4},
	pages = {1169--1181},
	year = {1981},
	month = oct,
	doi = {10.1143/PTP.66.1169},
	url = {https://doi.org/10.1143/PTP.66.1169},
}

@misc{sun2026foldedsurfacecodearchitecture,
      title={A Folded Surface Code Architecture for 2D Quantum Hardware}, 
      author={Zhu Sun and Zhenyu Cai},
      year={2026},
      eprint={2601.19823},
      archivePrefix={arXiv},
      primaryClass={quant-ph},
      url={https://arxiv.org/abs/2601.19823}, 
}

@misc{khalid2026impactsdecoderlatencyutilityscale,
      title={Impacts of Decoder Latency on a Utility-Scale Quantum Computer Architecture}, 
      author={Abdullah Khalid and Allyson Silva and Gebremedhin A. Dagnew and Tom Dvir and Oded Wertheim and Motty Gruda and Xiangzhou Kong and Mia Kramer and Zak Webb and Artur Scherer and Masoud Mohseni and Yonatan Cohen and Pooya Ronagh},
      year={2026},
      eprint={2511.10633},
      archivePrefix={arXiv},
      primaryClass={quant-ph},
      url={https://arxiv.org/abs/2511.10633}, 
}

@misc{fazio2025lowoverheadmagicstatecircuits,
      title={Low-overhead Magic State Circuits with Transversal CNOTs}, 
      author={Nicholas Fazio and Mark Webster and Zhenyu Cai},
      year={2025},
      eprint={2501.10291},
      archivePrefix={arXiv},
      primaryClass={quant-ph},
      url={https://arxiv.org/abs/2501.10291}, 
}

@misc{landahl2011faulttolerantquantumcomputingcolor,
      title={Fault-tolerant quantum computing with color codes}, 
      author={Andrew J. Landahl and Jonas T. Anderson and Patrick R. Rice},
      year={2011},
      eprint={1108.5738},
      archivePrefix={arXiv},
      primaryClass={quant-ph},
      url={https://arxiv.org/abs/1108.5738}, 
}

@ARTICLE{feldman2005using,
  author={Feldman, J. and Wainwright, M.J. and Karger, D.R.},
  journal={IEEE Transactions on Information Theory}, 
  title={Using linear programming to Decode Binary linear codes}, 
  year={2005},
  volume={51},
  number={3},
  pages={954-972},
  doi={10.1109/TIT.2004.842696}}

@article{zeng_transversality_2011,
	title = {{Transversality versus Universality for Additive Quantum Codes}},
	author = {Bei Zeng and Andrew Cross and Isaac L. Chuang},
	journal = {IEEE Transactions on Information Theory},
	volume = {57},
	number = {9},
	pages = {6272--6284},
	year = {2011},
	month = {sep},
	doi = {10.1109/tit.2011.2161917},
	url = {https://doi.org/10.1109/tit.2011.2161917},
	eprint = {0706.1382},
	archivePrefix = {arXiv},
	note = {arXiv:0706.1382},
}

@misc{knill_fault-tolerant_2004,
	title = {{Fault-Tolerant Postselected Quantum Computation: Schemes}},
	author = {E. Knill},
	year = {2004},
	publisher = {arXiv},
	url = {https://arxiv.org/abs/quant-ph/0402171},
	doi = {10.48550/arXiv.quant-ph/0402171},
	eprint = {quant-ph/0402171},
	archivePrefix = {arXiv},
	note = {arXiv:quant-ph/0402171 [quant-ph]},
}

@article{chamberland_fault-tolerant_2019,
	title = {{Fault-tolerant magic state preparation with flag qubits}},
	author = {Christopher Chamberland and Andrew W. Cross},
	journal = {Quantum},
	volume = {3},
	pages = {143},
	year = {2019},
	month = {may},
	doi = {10.22331/q-2019-05-20-143},
	url = {https://doi.org/10.22331/q-2019-05-20-143},
	eprint = {1811.00566},
	archivePrefix = {arXiv},
	note = {arXiv:1811.00566},
}

@article{derks_designing_2024,
	title = {{Designing fault-tolerant circuits using detector error models}},
	author = {Peter-Jan H. S. Derks and Alex Townsend-Teague and Ansgar G. Burchards and Jens Eisert},
	journal = {Quantum},
	volume = {9},
	pages = {1905},
	year = {2025},
	month = {nov},
	doi = {10.22331/q-2025-11-06-1905},
	url = {https://doi.org/10.22331/q-2025-11-06-1905},
	eprint = {2407.13826},
	archivePrefix = {arXiv},
	note = {arXiv:2407.13826},
}

@inproceedings{zhou_resource_2025,
	title = {{Resource Analysis of Low-Overhead Transversal Architectures for Reconfigurable Atom Arrays}},
	author = {Hengyun Zhou and Casey Duckering and Chen Zhao and Dolev Bluvstein and Madelyn Cain and Aleksander Kubica and Sheng-Tao Wang and Mikhail D. Lukin},
	booktitle = {Proceedings of the 52nd Annual International Symposium on Computer Architecture},
	pages = {1432--1448},
	year = {2025},
	month = {jun},
	doi = {10.1145/3695053.3731039},
	url = {https://doi.org/10.1145/3695053.3731039},
	eprint = {2505.15907},
	archivePrefix = {arXiv},
	note = {arXiv:2505.15907},
}

@misc{tuloup_computing_2026,
	title = {{Computing logical error thresholds with the Pauli Frame Sparse Representation}},
	author = {Thomas Tuloup and Thomas Ayral},
	year = {2026},
	publisher = {arXiv},
	url = {https://arxiv.org/abs/2603.14670},
	doi = {10.48550/arXiv.2603.14670},
	eprint = {2603.14670},
	archivePrefix = {arXiv},
	note = {arXiv:2603.14670 [quant-ph]},
}

@FOOTNOTE{note_frozen_gap_parameter,
    note = {We optimize the frozen gap parameter $\varphi$ via a parameter sweep.},
}

@FOOTNOTE{note_cluster_confidence_metric,
    note = {We employ the LLR 2-norm instance of the metric.},
}

@FOOTNOTE{fault_distance_4,
    note = {In reality, this protocol only has fault distance 4, as was noted in \cite{sahay2026fold}.}
}

@FOOTNOTE{scc_fd5_simulation,
note = {The color code cultivation circuits we simulate are sourced from the Zenodo associated with \cite{gidney_magic_2024}. The $f=3$ surface code cultivation circuits are from the GitHub repository for Ref. \cite{sahay2026fold}. The $f=5$ circuit, however, we reconstruct following \cite{sahay2026fold} and matching the noise conventions of the $f=3$ circuit, as it is not publicly available.
}
}

\clearpage
\onecolumngrid
\begin{center}
    \Large\textbf{End Matter}
\end{center}
\normalsize
\twocolumngrid

\textit{Algorithmic description of Caliper.---} Algorithm~\ref{alg:caliper} summarizes the search.
The distance table stores minimum path weights between graph nodes, including the spatial boundaries of the code. Caliper sums over all evaluated predictions in $\mathcal S$ to estimate $R$. For weights not known exactly, it uses the smallest correction weight found so far.

\begin{algorithm}[H]
    \caption{Caliper}\label{alg:caliper}
    \footnotesize
    \begin{algorithmic}[1]
        \Require Visible syndrome $\sigma_v$, distance table, limits $b,B$, depth and branching limits
        \State Decode $\sigma_v$ with hidden detectors free to obtain $\hat\sigma_{h_0}$ and $w_0(\hat\sigma_{h_0})$.
        \State Fix $\hat\sigma_{h_0}$ and decode the opposing logical class to obtain $w_1(\hat\sigma_{h_0})$.
        \State Initialize $\mathcal S \gets \{\hat\sigma_{h_0}\}$, $\mathcal Q \gets \{\hat\sigma_{h_0}\}$; $n\gets 0$.
        \While{$\mathcal Q\ne\varnothing$ and $n<B$}
            \State Take $\hat\sigma_h$ with lowest estimated $w_1$ from $\mathcal Q$.
            \State $\mathcal N \gets$ hidden syndromes differing from $\hat\sigma_h$ by one hidden detector value.
            \State Bound the matching weights using the distance table, then tighten the bounds by updating known corrections.
            \State Use MWPM as needed to determine the smallest weights.
            \State Store the weight estimates and bounds; $\mathcal S\gets\mathcal S\cup\mathcal N$.
            \State $n\gets n+1$.
            \State Select predictions in $\mathcal N$ whose bounds allow $w_1$ to stay the same or decrease, favoring low estimates of $w_1$ and respecting the depth and branching limits.
            \If{$n\leq b$ and $\mathcal N\ne\varnothing$}
                \State Also select the prediction in $\mathcal N$ with lowest estimated $w_1$.
            \EndIf
            \State Add selected predictions to $\mathcal Q$ only if not previously queued.
        \EndWhile
        \State \Return $R$ from Eq.~\eqref{seq:implemented-score}.
    \end{algorithmic}
\end{algorithm}

In the Supplementary Material, we detail the input parameters employed for the simulations in the main text, along with the effect of longer searches.

\textit{Latency.---} A critical aspect of any mid-circuit postselection metric is the speed at which
it can be computed. Its result directly determines the subsequent computation, so the quantum
computation must idle in the meantime \cite{fowler_time-optimal_2012, gidney_flexible_2019}. By using precomputed graph distances, Caliper achieves a
latency comparable to computing the complementary gap. Caliper is written primarily in \verb|C++|; for $f=3$ surface code cultivation, Caliper takes a median of $25.3\,\mu\mathrm{s}$ per shot, compared with
$15.5\,\mu\mathrm{s}$ for the closed-boundary complementary gap (measured over $10^4$ shots) on one core of
an Apple M4 processor. 

It can be further accelerated since, in practice, Caliper needs only to determine whether the score exceeds the acceptance threshold $R^*$.
Shots with an initial score sufficiently far below $R^*$ can often be discarded before exploring
further hidden syndromes. We set this early-rejection criterion from the observed changes in $R$ during
the search. Figure~\ref{fig:latency-ridgeline} compares the resulting decision latencies with those
of the thresholded closed-boundary gap.

\begin{figure}[t]
    \centering
    \includegraphics[width=0.8\linewidth]{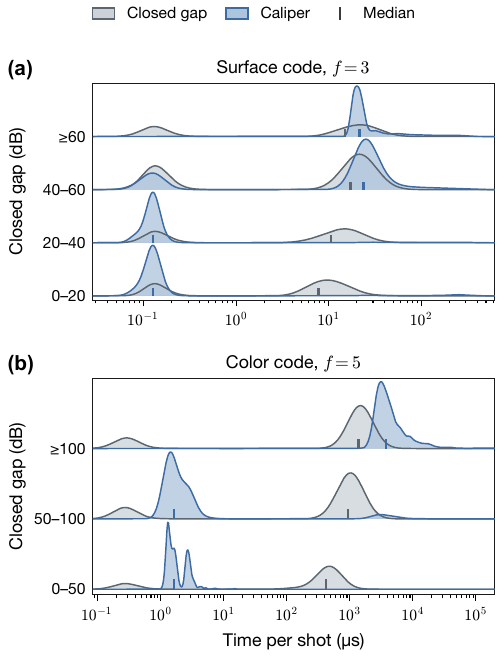}
    \caption{\textbf{Thresholded postselection latency.} Distributions of the time per shot for
    the closed-boundary complementary gap (gray) and Caliper (blue), grouped by the closed-boundary
    gap, for (a) $f=3$ surface code cultivation and (b) $f=5$ color code cultivation. Both methods
    determine whether their score meets an acceptance threshold. Vertical marks indicate the
    median within each group.}
    \label{fig:latency-ridgeline}
\end{figure}

On neutral-atom hardware, where a standard QEC cycle has a timescale on the order of milliseconds \cite{bluvstein_fault-tolerant_2026, computing2026quantumerrorcorrectiontoric}, Caliper applied to the $f=3$ protocols is fast enough to output a decision without additional idling. The larger protocols may require faster execution or a modest number of additional QEC cycles to bridge the reaction time \cite{gidney_flexible_2019, gidney_magic_2024}. On the microsecond timescales of superconducting hardware \cite{acharya_quantum_2025}, however, both the complementary gap—even with a closed boundary—and Caliper may require substantial acceleration, potentially using specialized hardware such as FPGAs \cite{das_scalable_2022, holmes_nisq_2020, liyanage_scalable_2023, wu_micro_2025}. In part, this is because, from our estimates, identifying the most likely class takes tens of microseconds, while the MWPM calculation for the complementary class can be up to $30\times$ slower. Without such acceleration, the spacetime cost of preserving the state while awaiting the postselection decision could dominate the cost of cultivation.

\textit{Calibration.---} An ideal postselection method is well-calibrated, which means that it can
accurately predict the posterior LER that will occur at whatever postselection strength we choose.
Calibration is not strictly necessary for a postselection method to perform well; clustering-based
metrics \cite{lee_efficient_2026}, for example, do not predict the post-selected LER but can still rank shots by how likely to
fail they seem. Nevertheless, calibration provides an illuminating diagnostic: it allows us to test
whether a method’s estimated confidence faithfully reflects the observed distribution and, when it
does not, to empirically investigate why.

In Fig.~\ref{fig:calibration}, we plot a reliability diagram for Caliper, along with closed-boundary
complementary gap. A perfectly calibrated metric would follow the $y=x$ line, where the predicted
and actual probabilities of logical failure are identical. Even in the closed-boundary setting, the
complementary gap appears miscalibrated, consistently underestimating logical error rates. The
reason for this overconfidence is well-understood: the predicted logical error rate is computed from
the complementary gap $G$ as $\hat P_L \approx e^{-G}/(1+e^{-G})$ \cite{gidney_yoked_2023, bombin_fault-tolerant_2024}, but it neglects the entropic factors
that describe the number of degenerate error mechanisms that produce gap $G$, along with higher-gap
mechanisms \cite{dennis_topological_2002, duclos-cianci_fast_2010}. At low but finite $p$, the contribution from higher-gap
mechanisms is negligible, leaving only the number of degenerate gap $G$ mechanisms, $N(G)$. The
closed-boundary calibration curve is well-described by a linear relationship on the log-log plot of
Fig.~\ref{fig:calibration}; a linear fit achieves $R^2 \simeq 0.97$ and slope $1-T$, where $T \simeq 0.18$
describes the effective (dimensionless) temperature of the system. For reference, the perfect calibration line corresponds to zero effective temperature, which is realized in the $p \rightarrow 0$ limit \cite{nishimori_internal_1981, dennis_topological_2002}.  

From the fit, we also obtain an empirical entropy model with prefactor $c_0$:
\begin{equation}
    \hat N(G) \approx c_0 e^{-(1-T)G}/e^{-G} = c_0 e^{TG}.
\end{equation}
This is consistent with our expectation that with a gap of $G$, the domain wall excitation
corresponds to a string of length $L\propto G$, and there are $N_L\propto \exp(L)$ such strings, so
the entropy of the system scales as $\Delta S \propto \log N(G) \propto G$ \cite{dennis_topological_2002}. Precisely in this sense,
as $T \rightarrow 0$, the closed-boundary curve is perfectly calibrated since the complementary
gap's zero-temperature approximation becomes exact.

\begin{figure}[t]
    \centering
    \includegraphics[width=0.8\linewidth]{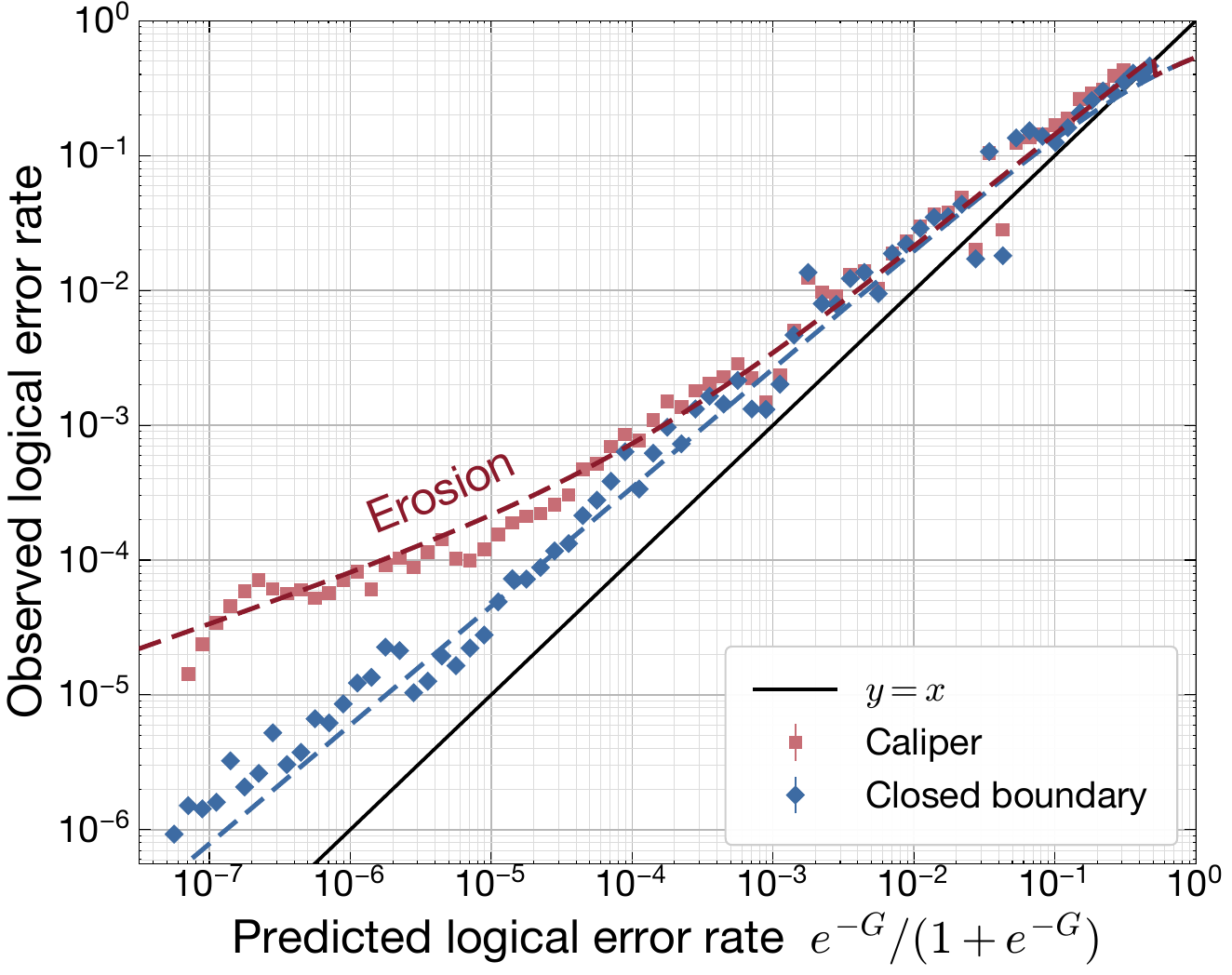}
    \caption{\textbf{Calibrations of Caliper and the complementary gap for $f=3$ surface code
    cultivation.} A power law fit, corresponding to a finite  temperature $T$, is applied to
    the closed boundary curve. Caliper is fit according to the phenomenological law
    \eqref{eq:phenomenological-description}.}
    \label{fig:calibration}
\end{figure}

For the open-boundary setting, Caliper consistently underestimates the logical error rate by an even
wider margin than the complementary gap. Moreover, its calibration curve is manifestly non-linear.
In part, this is because the failure probability conditioned on the open-boundary score $R$ depends
indirectly on $R$. By the tower property,
\begin{align}
    \Pr(\text{Fail} \mid R) & = \E_{G \sim \rho_{G \mid R}} [\Pr(\text{Fail} \mid G, R) \mid R]
\end{align}
This is averaged over, and thus depends on, the distribution of gaps $G$ given $R$, $\rho_{G \mid
R}$. To leading order, the conditional LER depends only on $G$: $\Pr(\text{Fail} \mid G, R) \approx
\Pr(\text{Fail} \mid G)$, where the only $R$-dependence in the conditional LER is through $\rho_{G
\mid R}$, and we numerically verify this behavior (SM). The non-linearity of the calibration curve
is also expected to arise from $\rho_{G \mid R}$.

Thus, we study $\rho_{G \mid R}$ in further empirical detail and summarize here the results; a more
thorough discussion can be found in the Supplementary Material. We observe that for the most part,
the closed-boundary complementary gap $G$ and Caliper score $R$ are close. However, in a select few instances, the
hidden syndrome reveals information that results in $G \ll R$: a phenomenon we call
\textit{erosion}. $\Delta = R - G$ measures the shortfall of the realized gap relative to the open-boundary
score, with larger $\Delta$ leading to larger discrepancies between the LER predicted
by Caliper and reality. It is even possible for a maximally large open-boundary score $R$ to erode to
a near-zero gap, at which point the decoding result is effectively random. That said, erosion is
rare. Outside of $R \sim G$, we find that the probability $R - G = \Delta$ is exponentially
suppressed in $\Delta$, where $\rho(G \mid R) \approx A e^{-\kappa \Delta}$. The eroded fraction of
shots at $R$ is
\begin{equation}
    E_R = \int_{0}^{R} \rho d\Delta = E_\infty (1 - e^{-\kappa R}),
    \label{eq:erosion-fraction}
\end{equation}
where $E_\infty=A/\kappa$ is the asymptotic eroded fraction in this model.
A shot at $R$ either did not erode and fails at a closed-boundary gap $R$ or erodes by $\Delta$ and
fails at a gap $R - \Delta$. Noting that $\rho(\Delta) \rightarrow 0$ exponentially quickly as
$\Delta \rightarrow \infty$, we obtain
\begin{equation}
    \label{eq:phenomenological-description}
    \Pr(\text{Fail} \mid R) \approx [1 - E_R] c_0 e^{-(1-T)R} + B e^{-\kappa R},
\end{equation}
where $T$ is the same temperature as in the closed-boundary setting and $B$ is the amplitude of
the contribution from eroded shots. We determine $\kappa$ from
the negative slope of a weighted linear fit to $\log\rho(\Delta)$ versus $\Delta$, excluding the
peak near $\Delta=0$. For each $\Delta$, we pool shots with $R\geq\Delta$; the fitted amplitude
gives $A$.

\clearpage
\onecolumngrid
\subsection*{\Large Supplementary Material}
\normalsize

\let\addcontentsline\oldaddcontentsline

\setcounter{equation}{0}
\setcounter{figure}{0}
\setcounter{table}{0}

\renewcommand{\theequation}{S\arabic{equation}}
\renewcommand{\thefigure}{S\arabic{figure}}
\renewcommand{\thetable}{S\arabic{table}}
\renewcommand{\refname}{}

\tableofcontents

\newpage

\section{Additional benchmarking}
\begin{figure}[!h]
    \centering
    \includegraphics[width=\linewidth]{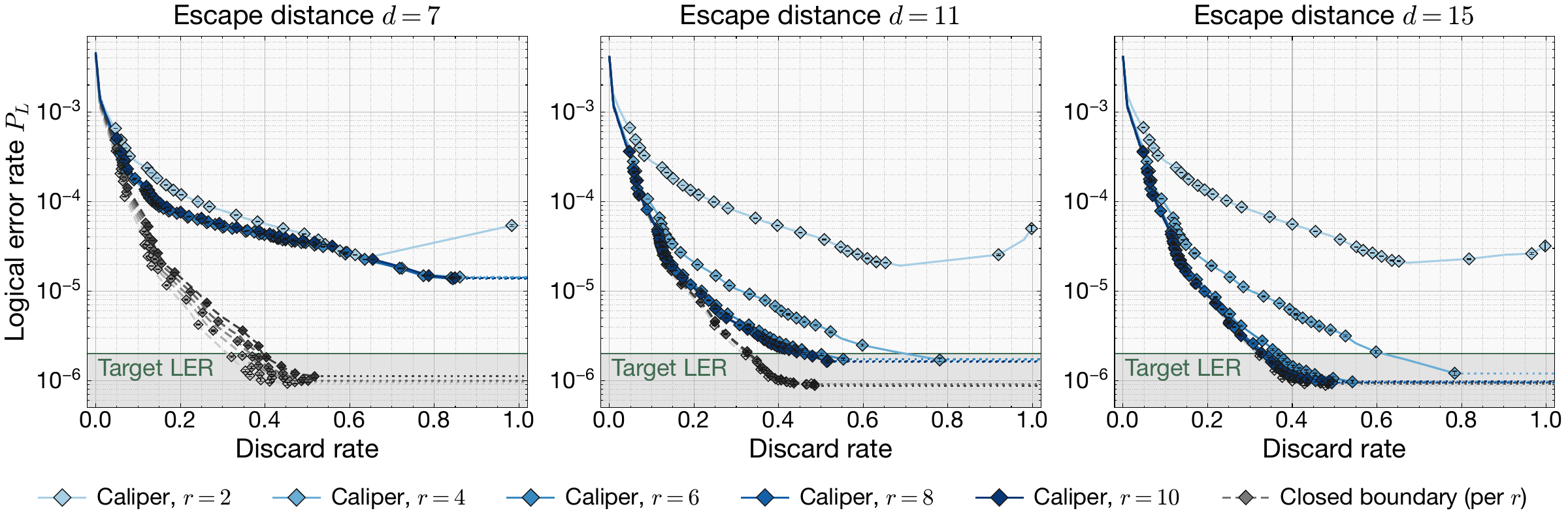}
    \caption{\textbf{Logical error rate versus discard rate for larger escape stages.} The curves
    for Caliper (blue) and closed-boundary postselection (gray) are depicted for escape code
    distances $d=7, 11, 15$ and idling round counts $r = 2, 4, \dots ,10$. }
    \label{fig:escape-sweep}
\end{figure}

In this section, we supply additional results from our numerical study of open-boundary
postselection, elaborating on the results detailed in the main text. Figure~\ref{fig:escape-sweep} shows how the logical error rate varies with the acceptance fraction for open-boundary postselection across modifications to the escape stage for $f=3$ surface code cultivation. Fig.~\ref{fig:modified-escape} is obtained by evaluating the logical error rate at the acceptance fraction corresponding to where closed-boundary postselection reaches the target logical error rate. As shown, the effects in Fig.~\ref{fig:modified-escape} due to insufficiently large escape stages persist at other acceptance fractions as well.

\begin{figure}[ht!]
    \centering
    \includegraphics[width=0.95\linewidth]{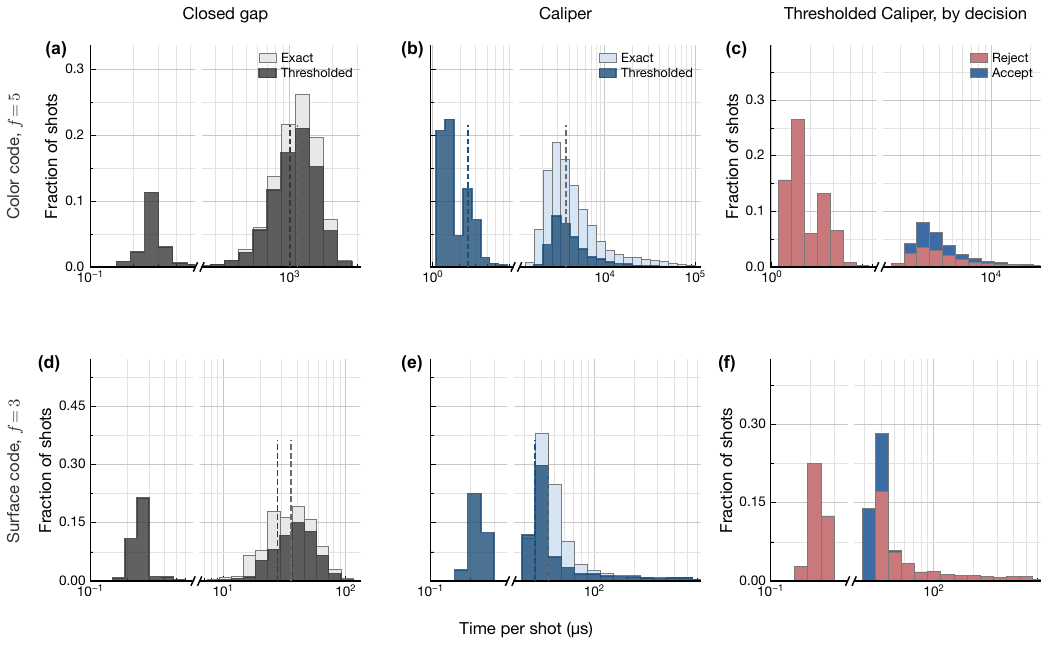}
    \caption{\textbf{Postselection decision latency.} (a) Frequency distribution of latencies for
    computing the closed boundary complementary gap in $f=5$ color code cultivation. Exact
    corresponds to evaluating this gap completely, whereas thresholded requires computing whether
    the gap is less than a threshold value $G^*$, calibrated here to the acceptance fraction that
    reaches the target LER of $10^{-9}$. (b) The corresponding plot for Caliper, where the
    score threshold $R^*$ is determined by the same acceptance fraction. (c) Breakdown of the thresholded
    curve by the decision to accept or discard the shot. (d)-(f) Extension of (a)-(c) to $f=3$ surface
    code cultivation, with the gap threshold determined by the target LER of $2 \times 10^{-6}$. }
    \label{fig:decision-latency}
\end{figure}
Fig.~\ref{fig:decision-latency} showcases the decoder latency for Caliper in further detail, evaluated in the same circumstances as in Fig.~\ref{fig:latency-ridgeline}. Without thresholding, Caliper is much slower than the closed-boundary complementary gap. This is because Caliper might spend considerable time searching the hidden syndrome space to estimate the open-boundary score more accurately, despite already having adequate information to conclude that a given shot should be discarded. The complementary gap is not accelerated to the same extent since it lacks Caliper's iterative, adaptive search structure. That being said, it may be possible to apply principles from thresholding to speed up estimates of gap-like quantities in non-matchable decoding settings \cite{xie_simple_2026, wills_forced_2026, landahl2011faulttolerantquantumcomputingcolor, feldman2005using}.

Next, we complement our discussion of erosion in the End Matter with additional numerics. There, we claimed that, to good
approximation, $P(\text{Fail} \mid G, R) \approx P(\text{Fail} \mid G)$. This is supported by
Fig.~\ref{fig:gap-contributions}, where we find that $P(\text{Fail} \mid G, R)$ is approximately
constant at fixed $G$ over a range of $R$. A higher-order effect arises in the neighborhood of $R =
G$, where we find that $P(\text{Fail} \mid G, R)$ decreases slightly; we defer studying its cause to
future work.

\begin{figure}[ht!]
    \centering
    \includegraphics[width=0.9\linewidth]{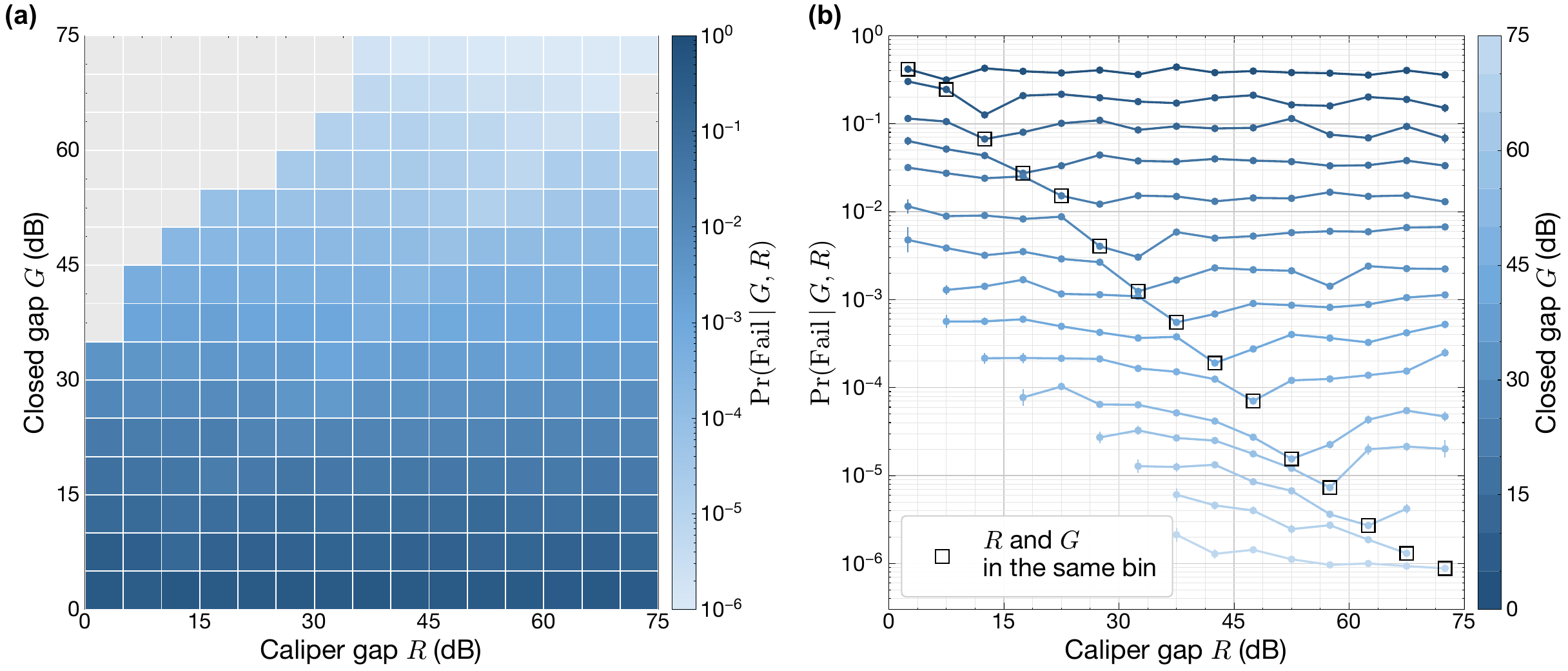}
    \caption{\textbf{Dependence of logical error rate on $R$ and $G$.} Data for panels correspond to
    $f=3$ surface code cultivation. (a) Logical error rate versus both the Caliper score $R$
    and the eventual closed boundary gap $G$. Gray squares correspond to $(R, G)$ pairs with fewer
    than 20 failures. (b) Sensitivity of the joint LER $P(\text{Fail} \mid G, R)$ to $R$ alone,
    binned by the eventual closed-boundary gap. Boxed points correspond to instances where $R$ and
    $G$ are in the same 5 dB bin.}
    \label{fig:gap-contributions}
\end{figure}
\begin{figure}[ht!]
    \centering
    \includegraphics[width=\linewidth]{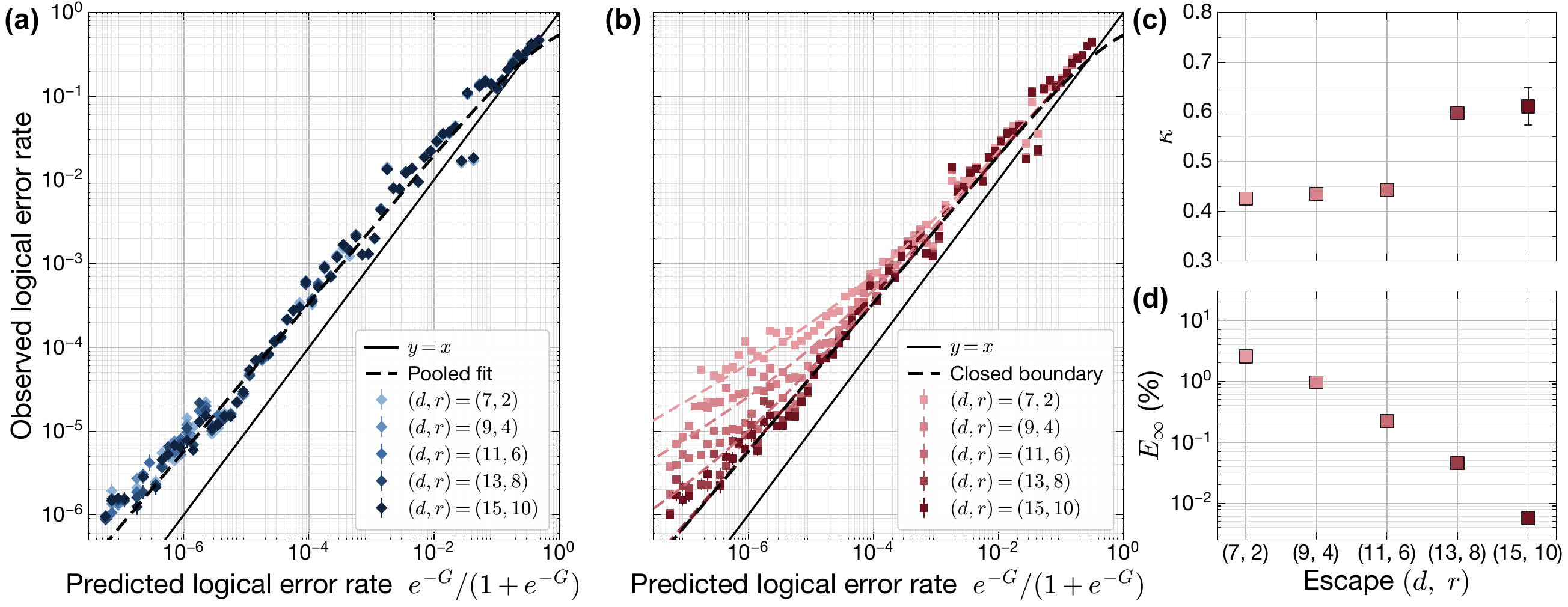}
    \caption{\textbf{Effect of modified escape stages on calibration.} Calibrations of the
    closed-boundary gap (a) and Caliper (b) as the number of idling rounds $r$ and distance of the
    escape code $d$ are increased. The linear fit shown in (a) is calculated by pooling the data
    across the different escape stages. The same line is illustrated in (b), along with the fit to
    the phenomenological description of erosion \eqref{eq:phenomenological-description} for each
    open-boundary curve. The decay rate $\kappa$ and asymptotic eroded fraction $E_\infty$ in
    Eq.~\eqref{eq:erosion-fraction} are extracted from fits to the erosion density for each escape
    stage independently and plotted against the escape stage size in (c) and (d).}
    \label{fig:modified-calibration}
\end{figure}

This effect in part is explained by Fig.~\ref{fig:modified-calibration}, where we find that
modifying the escape stage appears to have a negligible impact on the calibration curve for
closed-boundary postselection but considerably alters the phenomenology for open-boundary
postselection. The fact that closed-boundary postselection remains barely affected is consistent
with Ref. \cite{sahay2026fold}. There, after $r=3$ rounds of syndrome extraction in the $d=7$
surface code, the authors unitarily grow to a $d=13$ surface code and append a variable number of
idling rounds at that distance prior to readout, and they find that these additional idling rounds
have virtually no impact on the performance of cultivation. Our results indicate that this behavior
is not necessarily an accurate proxy for how altering the escape stage affects open-boundary
postselection. Fig.~\ref{fig:modified-calibration}(b) shows that for Caliper, the different escape
stages perform comparably at high logical error rates but they separate at low logical error rates.
They each follow the functional form of erosion described by
Eq.~\eqref{eq:phenomenological-description}, but with different parameters. In
Figs.~\ref{fig:modified-calibration}(c) and (d), we observe that $\kappa$, which describes how
quickly the effect of erosion drops off, increases as the escape stage grows in volume, and
$E_\infty$, the total probability of erosion, decreases. This model suggests that increasing the
size of the escape stage serves to suppress erosion until the bottleneck becomes the attainable
logical error rate even with the closed-boundary complementary gap.

Figure~\ref{fig:correlated-matching} compares vanilla and correlated matching for
$f=3$ surface code cultivation. Correlated matching lowers the logical error rate at fixed
discard rate for both the closed-boundary complementary gap and Caliper.
\begin{figure}[tb!]
    \centering
    \includegraphics[width=0.6\linewidth]{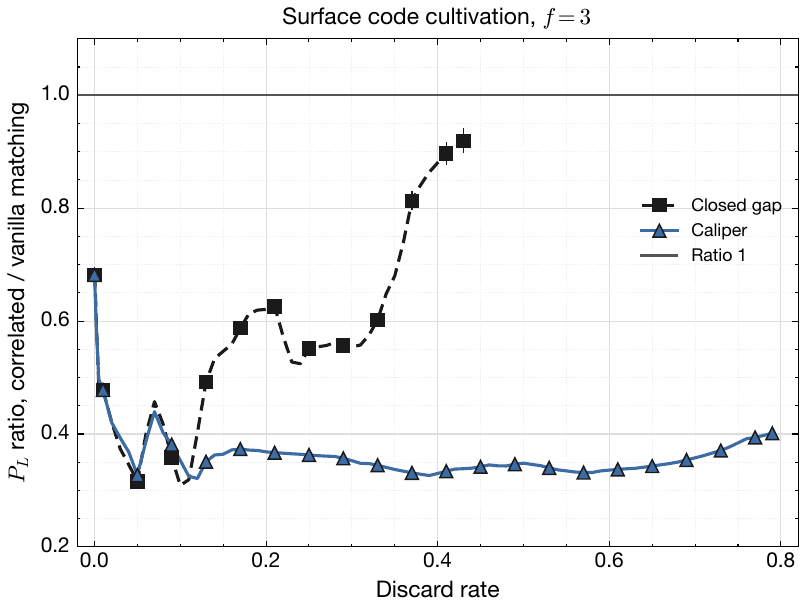}
    \caption{\textbf{Correlated matching.} Ratio of logical error rates with correlated and
    vanilla matching at fixed discard rate for the closed-boundary complementary gap and
    Caliper in $f=3$ surface code cultivation.}
    \label{fig:correlated-matching}
\end{figure}
Finally, Fig.~\ref{fig:search-depth-sweep} shows the effect of increasing Caliper's search depth.
The reductions in logical error rate are modest compared with the increase in evaluation time.

\begin{figure}[tb!]
    \centering
    \includegraphics[width=\linewidth]{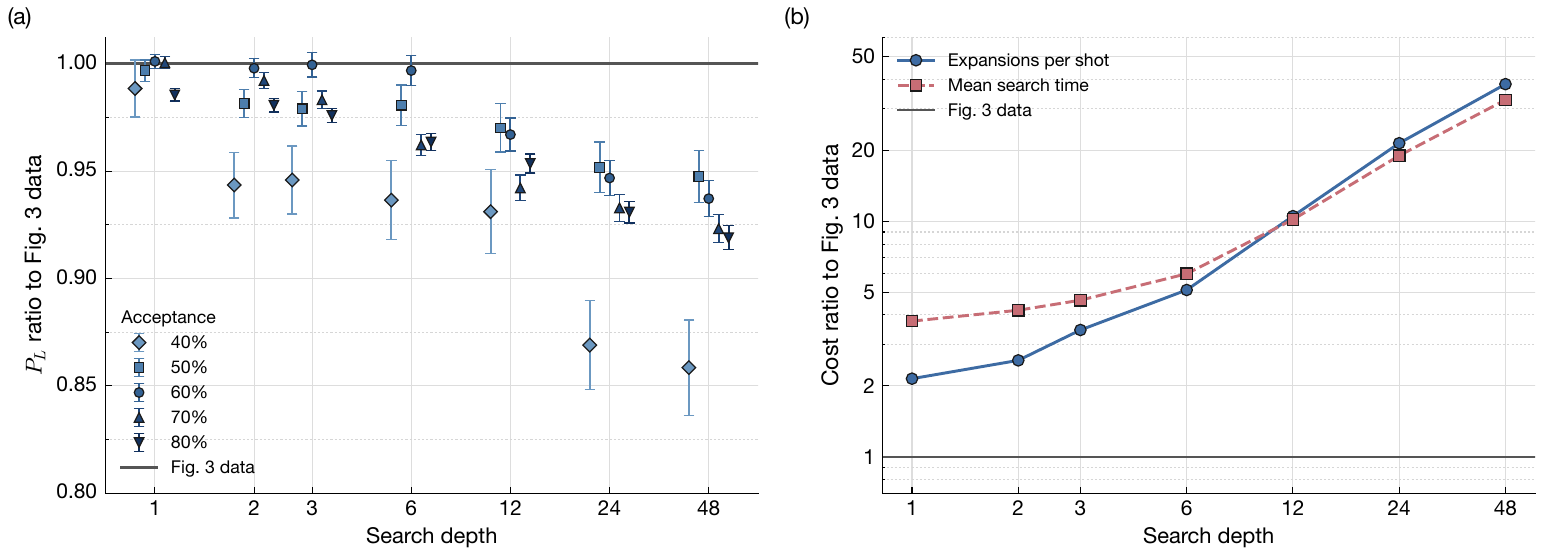}
    \caption{\textbf{Caliper search depth.} (a) Logical error rate at fixed acceptance,
    relative to the corresponding data in Fig.~\ref{fig:comparison}, as the search depth increases.
    (b) Mean number of expansions per shot and mean search time, normalized to the same reference.}
    \label{fig:search-depth-sweep}
\end{figure}

\FloatBarrier
\section{Guarantees from Caliper}
\label{sec:caliper}

In this section, we establish the properties of matching that allow Caliper to efficiently search
over possible hidden syndromes. Throughout, we work in the most-likely-error approximation used in
the main text. Exact maximum-likelihood decoding is $\#\mathrm{P}$-hard in general
\cite{iyer_hardness_2015}. As an efficient proxy, matching retains only the most likely error
configuration in each logical class.

\subsection{Matchable detector error models}

\begin{definition}[Detector error model]
    Consider a circuit subject to a set $\mathcal E$ of independent elementary error mechanisms.
    Each mechanism $\varepsilon\in\mathcal E$ occurs with probability $p_\varepsilon$ and flips a
    set of detectors $\partial\varepsilon$. It may also flip the logical observable, in which case
    we call $\varepsilon$ \emph{marked}. The collection of error mechanisms together with their
    detector syndromes, logical outcomes, and probabilities defines the detector error model (DEM)
    of the circuit.
\end{definition}

For an error configuration $E\subseteq\mathcal E$, the observed syndrome is
\begin{equation}
    \partial E
    =
    \bigoplus_{\varepsilon\in E}\partial\varepsilon ,
\end{equation}
and its logical outcome is the parity of the marked mechanisms in $E$. Since the mechanisms occur
independently,
\begin{equation}
    P(E)
    =
    \prod_{\varepsilon\in E}p_\varepsilon
    \prod_{\varepsilon\notin E}(1-p_\varepsilon)
    =
    \kappa e^{-w(E)},
\end{equation}
where $\kappa$ is independent of $E$ and
\begin{equation}
    w(E)
    =
    \sum_{\varepsilon\in E}\omega_\varepsilon,
    \qquad
    \omega_\varepsilon
    =
    \log\!\left(\frac{1-p_\varepsilon}{p_\varepsilon}\right).
    \label{seq:dem-weight}
\end{equation}
As in the main text, weights and scores use natural logarithms. Where figures quote values in dB,
the conversion is $10/\ln 10$ dB per natural-log unit. We assume the weights are positive.

\begin{definition}[Matchable DEM]
    A detector error model is \emph{matchable} if every elementary error mechanism flips at most two
    detectors,
    \begin{equation}
        |\partial\varepsilon|\leq 2
        \qquad
        \text{for all }\varepsilon\in\mathcal E .
    \end{equation}
\end{definition}

A matchable DEM has a weighted graph representation. Each detector is represented by a vertex. A
mechanism that flips two detectors $u$ and $v$ gives an edge between $u$ and $v$, while a mechanism
that flips a single detector is connected to a virtual boundary vertex. Boundary vertices carry no
detector constraint. Each edge has weight $\omega_\varepsilon$, and an edge is marked exactly when
the corresponding mechanism flips the logical observable. Parallel edges may be retained.

For a detector syndrome $S$, a \emph{correction} is an edge set $M$ whose odd-degree detector
vertices are precisely $S$; no parity constraint is imposed on virtual boundary vertices. Its weight
is
\begin{equation}
    \|M\|
    =
    \sum_{\varepsilon\in M}\omega_\varepsilon .
\end{equation}
Let $W_0(S)$ and $W_1(S)$ be the minimum weights of corrections with even and odd marked-edge
parity, respectively. We write
\begin{equation}
    w_0(S)=\min\{W_0(S),W_1(S)\},
    \qquad
    w_1(S)=\max\{W_0(S),W_1(S)\}.
\end{equation}
The complementary gap is therefore
\begin{equation}
    G(S)=w_1(S)-w_0(S).
\end{equation}

Because all edge weights are positive, a minimum-weight correction contains no unnecessary closed
loops. It may therefore be written as a collection of paths pairing syndrome defects either with one
another or with a boundary. Replacing every such path by a shortest path cannot increase the weight.
This gives the usual reduction to minimum-weight perfect matching.

It will be useful below to keep track of the logical parity of a path. Let
\begin{equation}
    d_b(u,v),\qquad b\in\{0,1\},
\end{equation}
denote the minimum weight of a path from $u$ to $v$ containing marked edges of parity $b$. These
distances depend only on the DEM and can therefore be precomputed. They can be obtained by running
shortest-path search on the graph in which each vertex is doubled according to the accumulated
marked-edge parity. When the logical parity is irrelevant, we write
\begin{equation}
    d(u,v)=\min\{d_0(u,v),d_1(u,v)\}.
\end{equation}

For a decoding graph with $V$ vertices and $E$ edges, the full distance table can be constructed by repeated Dijkstra searches in $O(V(V + E) \log V)$ time and stored in $O(V^2)$ memory. Tracking logical parity changes these costs only by constant factors. The table is computed once for each circuit and noise model and is reused across shots. 

\subsection{Changing a syndrome}

We now establish the basic observation used by Caliper. Suppose a syndrome $S$ is changed to
$S\oplus T$, where $T$ is a set of detector flips. Let $M$ be a minimum-weight correction for $S$.
For any edge set $J$, define its net cost relative to $M$ by
\begin{equation}
    c_M(J)
    =
    \sum_{\varepsilon\in J\setminus M}\omega_\varepsilon
    -
    \sum_{\varepsilon\in J\cap M}\omega_\varepsilon .
    \label{seq:net-cost}
\end{equation}
This is exactly the change in weight produced by replacing $M$ by $M\oplus J$.

\begin{lemma}[Exact change under a syndrome modification]
    \label{slem:increment}
    Let $M_\ell$ be a minimum-weight correction of syndrome $S$ in logical class $\ell$. Then
    \begin{equation}
        W_{\ell\oplus b}(S\oplus T)-W_\ell(S)
        =
        \min_J c_{M_\ell}(J),
        \label{seq:increment}
    \end{equation}
    where the minimum is over edge sets $J$ whose odd-degree detector vertices are exactly $T$ and
    whose marked-edge parity is $b$.
\end{lemma}

\begin{proof}
    For every allowed $J$, the edge set $M_\ell\oplus J$ has syndrome $S\oplus T$ and logical class
    $\ell\oplus b$. Hence
    \begin{equation}
        W_{\ell\oplus b}(S\oplus T)
        \leq
        W_\ell(S)+c_{M_\ell}(J).
    \end{equation}
    Minimizing over $J$ gives one inequality.

    Conversely, let $M'$ be a minimum-weight correction of $S\oplus T$ in class $\ell\oplus b$. Then
    \begin{equation}
        J=M_\ell\oplus M'
    \end{equation}
    has detector boundary $T$, marked-edge parity $b$, and
    \begin{equation}
        c_{M_\ell}(J)
        =
        \|M'\|-\|M_\ell\|
        =
        W_{\ell\oplus b}(S\oplus T)-W_\ell(S).
    \end{equation}
    Thus the bound is attained.
\end{proof}

The important special case is the relocation of a syndrome defect from $u$ to $v$, for which
$T=\{u,v\}$.

\begin{corollary}[Distance bound]
    \label{slem:distance-bound}
    For any $u,v$ and $b\in\{0,1\}$,
    \begin{equation}
        \left|
        W_{\ell\oplus b}(S\oplus\{u,v\})
        -
        W_\ell(S)
        \right|
        \leq
        d_b(u,v).
        \label{seq:distance-bound}
    \end{equation}
\end{corollary}

\begin{proof}
    Let $J$ be a minimum-weight $u$--$v$ path of marked-edge parity $b$. Toggling $J$ transforms any
    correction for $(S,\ell)$ into a correction for $(S\oplus\{u,v\},\ell\oplus b)$ and changes its
    weight by at most $\|J\|=d_b(u,v)$. Hence
    \begin{equation}
        W_{\ell\oplus b}(S\oplus\{u,v\})
        -
        W_\ell(S)
        \leq d_b(u,v).
    \end{equation}
    Applying the same path in the reverse direction gives the opposite inequality.
\end{proof}

Thus the effect of relocating a defect can always be bounded using a table lookup. The exact change
depends on how the path intersects the current correction, through Eq.~\eqref{seq:net-cost}, but its
magnitude is bounded independently of the shot.

Caliper frequently does better than this bound by explicitly considering the possible local
modifications of the current correction.

\begin{lemma}[Certificate for an exact local update]
    \label{slem:cert}
    Let $\mathcal J_{\rm loc}$ be a set of candidate modifications with the required syndrome and
    logical parity, and define
    \begin{equation}
        c_{\rm loc}
        =
        \min_{J\in\mathcal J_{\rm loc}} c_M(J).
    \end{equation}
    Suppose $\lambda$ is a lower bound on $c_M(J)$ for every allowed modification $J\notin\mathcal
    J_{\rm loc}$. If
    \begin{equation}
        c_{\rm loc}\leq\lambda,
    \end{equation}
    then
    \begin{equation}
        W_{\ell\oplus b}(S\oplus T)-W_\ell(S)
        =
        c_{\rm loc}.
    \end{equation}
\end{lemma}

\begin{proof}
    By Lemma~\ref{slem:increment}, the exact change is the minimum net cost over all allowed
    modifications. Every modification outside $\mathcal J_{\rm loc}$ has cost at least $\lambda$,
    while a modification inside $\mathcal J_{\rm loc}$ has cost $c_{\rm loc}\leq\lambda$. Hence the
    global minimum lies in $\mathcal J_{\rm loc}$.
\end{proof}

\begin{lemma}[Update after flipping one detector]
    \label{slem:chain}
    For the decoding graphs used here, let $M$ be a minimum-weight correction in one logical
    class, written as paths pairing defects or connecting them to a boundary. After flipping
    a hidden detector $u$, a minimum-weight
    correction in the same class can be found by reconnecting these pairs along one path
    starting at $u$ and ending at a boundary. The change in weight is the total weight of the
    new connections minus that of the removed ones.
\end{lemma}
\begin{proof}
    Compare the pairs in $M$ with those in a minimum-weight correction after the flip. The
    connections appearing in only one correction form a path from $u$ to a boundary, together
    with closed loops and paths between boundaries. Each remaining loop or boundary-to-boundary
    path preserves all detector values and the logical class. Its net cost must be zero: a negative
    cost would improve $M$, while a positive cost would improve the new correction when
    reversed. We can therefore omit these loops and paths, leaving only the path starting at $u$.
\end{proof}

Caliper explores these reconnections using the precomputed distances. It reuses calculations
for the same partial update and stops exploring a choice once its lower bound exceeds the best
update found so far. Once every possible update has been checked or ruled out by the bounds,
the matching weight is exact. The calculation is repeated for both logical classes, then the
two weights are ordered as $w_0\leq w_1$.

If this calculation is too large or reaches its search limit, Caliper continues with bounds
on the weights. It improves these bounds where possible and uses MWPM as needed to identify
the smallest weights. A further decode may also be needed to construct the correction used
in the next search step. Some weights can remain approximate.

\textit{Simulation settings.---} In the main text, we use $b=6$ and $B=18$, select up to ten
neighbors by the non-increasing rule, and apply this rule only when expanding a prediction
within three flips of the initial prediction. The exact update calculation uses limits of
63 defects, 57 matched pairs, and successive limits of $2048$, $32768$, and $524288$ partial
updates. The margins for additional decoding and for correction updates are $2.5$ and $5$,
respectively, in natural-log units.

There is a particularly useful case in which no change to the gap is possible.

\begin{lemma}[Cancellation away from the critical path]
    \label{slem:cancel}
    Let $M_0$ and $M_1$ be minimum-weight corrections for syndrome $S$ in the two logical classes.
    Suppose a class-preserving modification $J$ is certified by Lemma~\ref{slem:cert} to be optimal
    for both class corrections. If
    \begin{equation}
        J\cap(M_0\oplus M_1)=\varnothing,
        \label{seq:no-critical-intersection}
    \end{equation}
    then
    \begin{equation}
        G(S\oplus T)=G(S).
    \end{equation}
\end{lemma}

\begin{proof}
    Equation~\eqref{seq:no-critical-intersection} implies that every edge of $J$ either belongs to
    both $M_0$ and $M_1$ or to neither. Therefore
    \begin{equation}
        c_{M_0}(J)=c_{M_1}(J).
    \end{equation}
    Since $J$ preserves the logical class and is optimal for both modified matching problems,
    Lemma~\ref{slem:increment} shows that $W_0$ and $W_1$ change by exactly the same amount. Their
    difference is therefore unchanged.
\end{proof}

The symmetric difference $M_0\oplus M_1$ is the part of the matching that distinguishes the two
logical classes. Lemma~\ref{slem:cancel} therefore formalizes the intuition used by Caliper:
hidden-syndrome modifications that do not interact with this region can often be shown to leave the
complementary gap exactly unchanged.

\subsection{Minimum-weight limit of the open-boundary score}

We now apply these facts to the open-boundary problem. At the time of the postselection decision,
only a subset of the detector outcomes is known. Let $\sigma_v$ denote this \emph{visible syndrome}.
The remaining detector outcomes, which would only be revealed by the later closing measurement, form
the \emph{hidden syndrome} $\sigma_h$.

\begin{definition}[Hidden-syndrome prediction]
    A \emph{hidden-syndrome prediction} $\hat\sigma_h$ is a possible assignment of outcomes to all
    hidden detectors. Together with the observed syndrome $\sigma_v$, it determines a prediction
    of the full syndrome,
    \begin{equation}
        \hat\sigma=\sigma_v\oplus\hat\sigma_h.
    \end{equation}
\end{definition}

Different predictions correspond to different possible syndromes that could be revealed if the
circuit were subsequently closed by the final measurement. For a prediction $\hat\sigma_h$, let
\begin{equation}
    w_0(\sigma_v\oplus\hat\sigma_h)
\end{equation}
and
\begin{equation}
    w_1(\sigma_v\oplus\hat\sigma_h)
\end{equation}
denote the smaller and larger minimum-weight logical-class corrections for the resulting full
syndrome.

Allowing the hidden detector values to vary freely gives
\begin{equation}
    w_0^*
    =
    \min_{\hat\sigma_h}
    w_0(\sigma_v\oplus\hat\sigma_h).
    \label{seq:omega-prime}
\end{equation}
This is simply the weight of the most likely error consistent with the information available at
decision time, and it is obtained by decoding with an open time boundary.

The corresponding minimum-weight limit of the open-boundary score is
\begin{equation}
    R_{\min}(\sigma_v)
    =
    \min_{\hat\sigma_h}
    w_1(\sigma_v\oplus\hat\sigma_h)
    -
    w_0^*.
    \label{seq:open-gap}
\end{equation}
The first minimization asks for the hidden syndrome under which an error in the less
likely logical class can occur with the smallest weight. Importantly, the hidden syndrome attaining this
minimum need not be the same hidden syndrome that attains $w_0^*$. Thus simply computing the two
unrestricted open-boundary matching weights does not determine Eq.~\eqref{seq:open-gap}; one must
search over possible hidden syndromes.

The open-boundary decode selects an initial prediction $\hat\sigma_{h_0}$ attaining $w_0^*$.
Caliper fixes this prediction and decodes the opposing class, then searches outward by changing
hidden detector values. The bounds above apply to each neighboring hidden syndrome; they are not lower
bounds on every hidden syndrome that could be reached by further changes.

\subsection{Guarantee for the Caliper search}

Let $\mathcal S$ be the set of predictions evaluated by the search, including its root, and define
\begin{equation}
    R_{\min}^{(\mathcal S)}(\sigma_v)
    = \min_{\hat\sigma_h\in\mathcal S}w_1(\sigma_v\oplus\hat\sigma_h)-w_0^*.
    \label{seq:truncated-gap}
\end{equation}

\begin{theorem}[Minimum over evaluated predictions]
    \label{sthm:main}
    If the class weights stored for each prediction in $\mathcal S$ are exact, then
    \begin{equation}
        R_{\min}(\sigma_v)\leq R_{\min}^{(\mathcal S)}(\sigma_v)
        \leq G(\sigma_v\oplus\hat\sigma_{h_0}).
        \label{seq:one-sided-depth}
    \end{equation}
    Adding evaluated predictions can only decrease $R_{\min}^{(\mathcal S)}$.
\end{theorem}
\begin{proof}
    The minimum over a subset cannot be smaller than the unrestricted minimum, and the root is
    included in $\mathcal S$. Enlarging the set cannot increase its minimum.
\end{proof}

We note that Caliper can find the smallest $w_1$ in $\mathcal S$ without pinpointing every
weight exactly. Only predictions whose lower bounds are below the smallest exact weight
found so far could improve that minimum. Caliper tightens these bounds or decodes the
predictions until none can do so.

This guarantee concerns the minimum opposing weight among evaluated predictions. Caliper limits
which predictions it expands. A prediction rejected for
expansion may have a descendant with lower opposing weight, so these local decisions do not certify
the global minimum or exhaustive coverage within a fixed number of flips. Even when the frontier
is exhausted, a restricted admission rule need not visit every hidden syndrome. Exhaustive enumeration would
recover $R_{\min}$, but would require exponentially many hidden syndromes. Nor does the monotonicity of
$R_{\min}^{(\mathcal S)}$ imply monotonicity of the summed score $R$ below.

\subsection{Relation to the logical failure probability}

Finally, we connect the minimum in Eq.~\eqref{seq:open-gap} to the most-likely-error approximation
of the open-boundary logical failure probability.

For a set $\mathcal S$ of hidden-syndrome predictions, define
\begin{equation}
    A_{\mathcal S}
    =
    \sum_{\hat\sigma_h\in\mathcal S}
    e^{-w_0(\sigma_v\oplus\hat\sigma_h)}
\end{equation}
and
\begin{equation}
    B_{\mathcal S}
    =
    \sum_{\hat\sigma_h\in\mathcal S}
    e^{-w_1(\sigma_v\oplus\hat\sigma_h)}.
\end{equation}
Within the most-likely-error approximation, $A_{\mathcal S}$ is the total weight assigned to the
more likely logical class of each predicted full syndrome, while $B_{\mathcal S}$ is the corresponding
weight of the less likely class. The logical failure probability restricted to these predictions is
therefore
\begin{equation}
    \widehat P_{\rm L}^{(\mathcal S)}(\sigma_v)
    =
    \frac{B_{\mathcal S}}
    {A_{\mathcal S}+B_{\mathcal S}}.
    \label{seq:restricted-failure}
\end{equation}

It is convenient to define
\begin{equation}
    R_{\mathcal S}(\sigma_v)
    =
    \log(A_{\mathcal S}+B_{\mathcal S})-\log B_{\mathcal S},
    \label{seq:aggregate-score}
\end{equation}
for which
\begin{equation}
    \widehat P_{\rm L}^{(\mathcal S)}(\sigma_v)
    =
    e^{-R_{\mathcal S}}.
    \label{seq:aggregate-probability}
\end{equation}

When $\mathcal S$ contains every possible hidden syndrome,
Eq.~\eqref{seq:restricted-failure} is the full most-likely-error approximation to the open-boundary
logical failure probability. Caliper instead sums over the predictions it has evaluated, using the smallest correction weights found during the search. Both sums omit positive terms, so the
resulting score can err in either direction. The alternative log-odds score
$\log(A_{\mathcal S}/B_{\mathcal S})$ orders shots identically to $R_{\mathcal S}$ and approaches it
when the estimated failure probability is small.

For each prediction, Caliper keeps the smallest correction weight found in each logical class.
We write these as $U_0\leq U_1$, sorting them in the same way as $w_0,w_1$. Since these corrections
need not have minimum weight, $U_i\geq w_i$. Caliper computes
\begin{equation}
    R=-\log\frac{\sum_{\hat\sigma_h\in\mathcal S}e^{-U_1(\hat\sigma_h)}}
    {\sum_{\hat\sigma_h\in\mathcal S}
    \left[e^{-U_0(\hat\sigma_h)}+e^{-U_1(\hat\sigma_h)}\right]}.
    \label{seq:implemented-score}
\end{equation}
In the exact limit, where both matching weights are known exactly for every prediction in
$\mathcal S$, $R_{\mathcal S}=R$.

The connection with the minimum-weight score follows from the usual log-sum-exp limit. If all edge
weights are scaled by a common factor $\beta$, then
\begin{equation}
    \lim_{\beta\rightarrow\infty}
    \frac{R_{\mathcal S}}{\beta}
    =
    \min_{\hat\sigma_h\in\mathcal S}
    w_1(\sigma_v\oplus\hat\sigma_h)
    -
    \min_{\hat\sigma_h\in\mathcal S}
    w_0(\sigma_v\oplus\hat\sigma_h).
    \label{seq:zero-temp}
\end{equation}
For the unrestricted set of hidden syndromes, the right-hand side is exactly $R_{\min}(\sigma_v)$.

Equations~\eqref{seq:open-gap} and \eqref{seq:zero-temp} distinguish the two quantities.
The open-boundary decoder immediately determines the minimum possible $w_0$, while the search
finds predictions with low $w_1$. The minimum-weight quantity $R_{\min}^{(\mathcal S)}$ obeys
Theorem~\ref{sthm:main}; the finite-sum score $R$ estimates the probability of eventual failure.
Neither is generally ordered relative to the closed-boundary gap of the subsequently observed
hidden syndrome.

\section{Observables and temporal boundaries}
We discuss here the necessity of a measurement that closes a temporal boundary for a useful
definition of the complementary gap. In particular, we demonstrate that logical observables defined
mid-circuit are inequivalent to observables defined on temporal boundaries, and they do not
implement useful realizations of the complementary gap.

To do so, we consider the spacetime code associated with the circuit-level DEM. In complete
generality, it can be formally described by a subsystem code \cite{pesah_fault-tolerant_2025}; we
omit the explicit formulation here since it is not necessary for our purposes. The relevant
consideration is the role of detector correlations as circuit stabilizers. We should first note that
there are two distinct flavors of circuit stabilizers, which should be distinguished. The first are
the space-time stabilizers that correspond to error mechanisms that are undetectable but also do not
flip any logical observable. These can be used to deform any logical error into all of its
degenerate representatives in the circuit. The other set of stabilizers corresponds to deterministic
correlations in the circuit that do not reveal any logical information, which are precisely
the detectors of the circuit. The logical correlations, on the other hand, are the membranes
defined in spacetime that capture the propagation of logical information through the circuit. In the
case of our cultivation circuits, these are $\bar Y$ membranes that extend backwards from the
terminal measurement.

Typically, detectors are evaluated from the perspective of parity checks in the circuit, but they
also function as stabilizers. They can be used to deform certain logical correlations to other,
logically equivalent correlations \cite{derks_designing_2024}. This equates to taking any representative membrane for $\bar Y$
and dressing it with the support of any set of detectors (Fig.~\ref{fig:membrane-equivalence}). Conversely, if two such membranes are
indeed logically equivalent correlations, their product must therefore be a detector. The
implication of this argument is that there cannot be any undetectable logical error that
flips---intersects an odd number of times---one observable without flipping the other. If there
were, taking the product of them would imply the resulting detector is flipped by the error, which
contradicts that it is undetectable.

\begin{figure}[tb!]
    \centering
    \includegraphics[width=\linewidth]{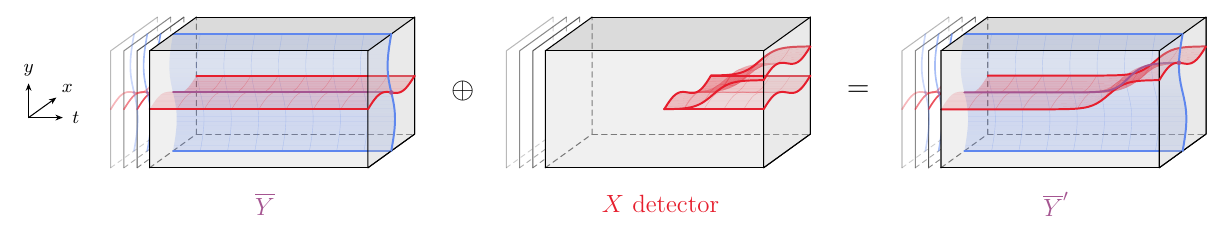}
    \caption{\textbf{Equivalent logical membranes.} Adding an $X$ detector deforms $\bar Y$ into
    an equivalent membrane $\bar Y'$, retaining support on the terminal measurement boundary.}
    \label{fig:membrane-equivalence}
\end{figure}

From this, one can deduce that a logical observable defined instantaneously in the middle of a noisy
circuit cannot be equivalent to a logical operator defined on the temporal boundary of the final
measurement. The reason for this is that there are undetectable logical errors that can be applied
between these two times, which seem to flip one logical correlation but not the other. Thus, the
$\bar Y$ observable defined in the middle of a noisy circuit backpropagates to a membrane in
spacetime that is \emph{inequivalent} to the observable defined on the terminal boundary. Another
way to see this is to explicitly consider deformations of the $\bar Y$ membrane backpropagating from
the terminal boundary. One can verify that there is no way to clean the membrane off the terminal
boundary, since there is no product of detectors whose support on the final measurement can
evaluate to a logical operator.

This is precisely the sense in which the complementary gap is not well-defined when naively applied
to observables defined in the middle of the circuit. In particular, an observable defined at time
$t$ in the middle of the circuit corresponds to a spacetime membrane extending backwards in time
from $t$. This membrane can be flipped by the
first type of stabilizers we described, as shown in Fig.~\ref{fig:midcircuit-membrane}. Of course, those stabilizers do not flip
the membranes of the correctly defined observables. As many of these stabilizers have weight $O(1)$,
if we tried to run the complementary gap on observables defined in the middle of the circuit, we
would find that every shot has a small, constant-size gap, and would not achieve the desired
suppression.

\begin{figure}[tb!]
    \centering
    \includegraphics[width=0.6\linewidth]{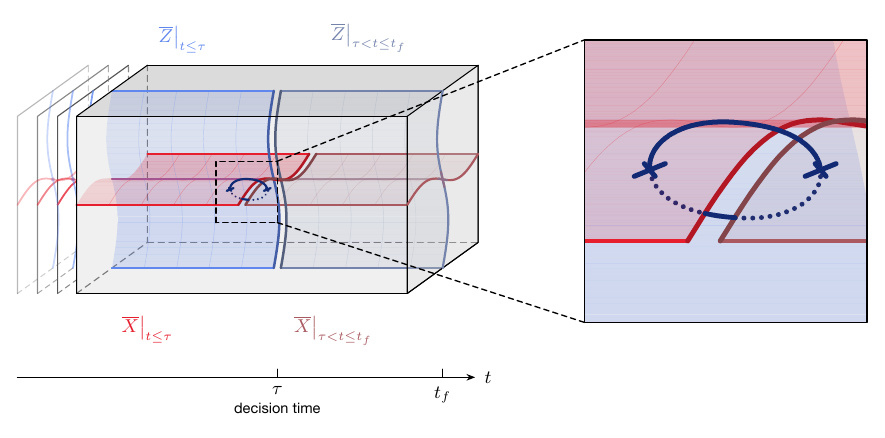}
    \caption{\textbf{Logical membranes at an open temporal boundary.} The local fault loop
    (inset) intersects the $\bar X$ membrane once before the decision time $\tau$ and once after,
    flipping the truncated observable while preserving the full observable at $t_f$.}
    \label{fig:midcircuit-membrane}
\end{figure}

Note that if we have a noiseless syndrome extraction round, it becomes possible to clean $\bar Y$
off the terminal boundary and to this noiseless error correction round. This is because, with the
assumption of no further errors, there are many additional, linearly independent correlations that
are assumed to be deterministic. We emphasize that these emerge artificially from the assumption
that there is no noise in the system, rather than a result of error correction. It would be atypical
to call these additional correlations detectors, but they function as circuit-level stabilizers just
like the detectors. With these added correlations, it becomes immediate to clean $\bar Y$ to the
preceding perfect syndrome extraction round, since without any subsequent noise, the two observables must agree. 

\FloatBarrier
\section{Optimality of Full Error Detection}
In this section, we discuss the optimality of full-error detection on the visible syndrome in
mid-circuit postselection. This conclusion is often implicitly assumed by carrying over the
intuition from closed-boundary postselection. In the closed-boundary setting, one decides whether to
retain the classical statistics from an observed shot based on its complete syndrome. The strategy
that retains only shots with trivial syndrome, $\sigma=0$, is optimal.

One way to understand this result in the closed-boundary setting is through the
statistical-mechanical model associated with the corresponding space-time code. For a fixed observed
syndrome $\sigma$, the compatible fault configurations decompose into logical equivalence classes.
If $Z_\ell(\sigma)$ denotes the partition function obtained by summing the probabilities of all
fault configurations with syndrome $\sigma$ and logical class $\ell$, then
\begin{equation}
    \Pr(\ell\mid \sigma)
    =
    \frac{Z_\ell(\sigma)}
    {\sum_{\ell'} Z_{\ell'}(\sigma)}.
\end{equation}
Thus the posterior logical error probability is determined by the relative free energies
\begin{equation}
    F_\ell(\sigma)=-\beta^{-1}\log Z_\ell(\sigma)
\end{equation}
of the different logical sectors. In particular, after maximum-likelihood decoding,
\begin{equation}
    p_{\mathrm L}(\sigma)
    =
    1-
    \frac{\max_\ell Z_\ell(\sigma)}
    {\sum_{\ell'}Z_{\ell'}(\sigma)}.
\end{equation}
At $p/p_{th} \ll 1$, i.e., low temperatures, these free-energy differences are dominated by the
corresponding minimum energies, recovering the familiar picture in which logical reliability is
controlled by the energy gap between the most likely logical sector and its competitors.

The open-boundary problem is subtler. In a cultivation-like circuit, the complete space-time
syndrome is not available: a subset of detector outcomes is absent because the circuit terminates at
an open temporal boundary. Consequently, postselection can depend only on the visible syndrome
$\sigma_v$, and the posterior logical error probability must average over the unobserved portion of
the syndrome. The usual closed-boundary argument therefore does not by itself imply that
$\sigma_v=0$ is optimal -- it is not necessarily the case that post-selecting purely on $\sigma_v =
0$ shots minimizes the logical error rate.

Below we argue that at sufficiently low noise, full error-detection on the visible syndrome is
nonetheless approximately optimal in a sense that is made precise below.

\begin{theorem}[Low-noise optimality of zero visible syndrome]

    Let \(v\) and \(h\) denote the visible and hidden parts of the syndrome. For \(\ell\in\{0,1\}\),
    let \(d_\ell(v,h)\) be the minimum number of faults producing syndrome \((v,h)\) and logical
    class \(\ell\). Define
    \begin{equation}
        a(v)=\min_{h,\ell} d_\ell(v,h),
        \qquad
        b(v)=\min_h \max_{\ell} d_\ell(v,h),
    \end{equation}
    and \(g(v)=b(v)-a(v)\).

    If the complete syndrome \((v,h)\) is decoded by maximum likelihood, then the conditional
    probability of a logical error given $v$ is
    \begin{equation}
        \Pr(\mathrm{Fail}\mid v)
        =
        \Theta\!\left(p^{g(v)}\right),
    \end{equation}
    and
    \begin{equation}
        g(v)\leq g(0)
    \end{equation}
    for every attainable visible syndrome \(v\).

\end{theorem}

\begin{proof}
    Let \(P_\ell(v,h)\) be the probability that the syndrome is \((v,h)\) and the logical class is
    \(\ell \in \{0,1\}\). The contributions to \(P_\ell(v,h)\) are dominated by minimum weight
    faults, so that
    \begin{equation}
        P_\ell(v,h)
        =
        \Theta\!\left(p^{d_\ell(v,h)}\right).
    \end{equation}

    For a fixed full syndrome \((v,h)\), maximum-likelihood decoding fails with joint
    probability
    \begin{equation}
        \min\{P_0(v,h),P_1(v,h)\}
        =
        \Theta\!\left(
        p^{\max\{d_0(v,h),d_1(v,h)\}}
        \right).
    \end{equation}
    Summing over \(h\) and approximating the sum by the dominant (minimum weight) contribution
    therefore gives
    \begin{equation}
        \Pr(v,\mathrm{Fail})
        =
        \Theta\!\left(p^{b(v)}\right).
    \end{equation}
    where entropic factors are absorbed into the big-$\Theta$ notation. Similarly,
    \begin{equation}
        \Pr(v)
        =
        \Theta\!\left(p^{a(v)}\right).
    \end{equation}
    Hence
    \begin{equation}
        \Pr(\mathrm{Fail}\mid v)
        =
        \Theta\!\left(p^{b(v)-a(v)}\right).
    \end{equation}

    It remains to compare \(g(v)\) with \(g(0)\). Let \(E\) be a fault configuration of weight
    \(a(v)\) producing visible syndrome \(v\). Choose two fault configurations \(A_0\) and \(A_1\)
    that attain \(b(0)\). Thus \(A_0\) and \(A_1\) have zero visible syndrome, the same hidden
    syndrome, opposite logical classes, and at most \(b(0)\) faults each.

    The configurations
    \begin{equation}
        E\mathbin{\oplus}A_0,
        \qquad
        E\mathbin{\oplus}A_1,
    \end{equation}
    where $\oplus$ denotes symmetric difference, have visible syndrome \(v\), the same hidden
    syndrome, and opposite logical classes. Each has at most \(a(v)+b(0)\) faults. Therefore
    \begin{equation}
        b(v)\leq a(v)+b(0).
    \end{equation}
    Since the no-fault configuration produces the zero syndrome, \(a(0)=0\), and hence
    \(g(0)=b(0)\). It follows that
    \begin{equation}
        g(v)=b(v)-a(v)\leq b(0)=g(0).
    \end{equation}
\end{proof}

Consequently, if
\begin{equation}
    g(v)<g(0)
    \qquad
    \text{for every }v\neq 0,
\end{equation}
then, for all sufficiently small \(p\),
\begin{equation}
    \Pr(\mathrm{Fail}\mid v)
    >
    \Pr(\mathrm{Fail}\mid 0)
    \qquad
    \text{for every }v\neq0.
\end{equation}
In this case, keeping only shots with \(v=0\) is the optimal postselection rule at its corresponding
acceptance probability. However, it is possible that $g(v)=g(0)$, in which case it remains possible
that full error-detection on the visible syndrome leads to a larger conditional probability of a
logical error by a constant factor than a postselection policy that retains some shots with a
non-trivial visible syndrome, even as $p\rightarrow 0$. Moreover, at larger $p$, for example in
near-threshold regimes, it is possible that other postselection strategies may be optimal.

\section{Additional effects from \texorpdfstring{$\ket{\bar T}$}{T} cultivation}
\label{app:t-vs-y}
In this work we simulated the Clifford circuit for $\ket{\bar Y}$ cultivation with open-boundary escape
stage postselection as a proxy for $\ket{\bar T}$ state cultivation, which allows for a direct comparison
with the closed-boundary postselection $\ket{\bar Y}$ state cultivation results of Refs.
\cite{sahay2026fold, gidney_magic_2024}. Recent works \cite{li2025soft, chase_clifft_2026,
takada2026exactefficientsimulationnoisy, hartweg_comparing_2026, wan_exact_2026,
chan_diagnosing_2026} have developed numerical methods to simulate the exact non-Clifford circuits
of $\ket{\bar T}$ cultivation. Simulating only the injection and cultivation stages (without the escape
stage), Refs. \cite{li2025soft, chase_clifft_2026, wan_exact_2026} identified that the logical error
rate of $f=5$ color code cultivation is roughly $15 \times$ larger than that obtained from
$\ket{\bar Y}$ cultivation proxies, as opposed to the $2 \times$ originally expected
\cite{gidney_magic_2024}. These discrepancies can be understood as originating from hook errors in
the non-Clifford circuit. Ref. \cite{chan_diagnosing_2026} recently introduced a flagging method to
recover the pre-escape circuit-distance with a modest increase in the expected number of retries.

However, Ref. \cite{sahay2026fold} studied the discrepancy between $\ket{\bar Y}$ state and $\ket{\bar T}$
state cultivation including the escape stage (with closed-boundary postselection) for the $f=3$
protocol. The authors conclude that the escape stage dominates the error-rate of the entire protocol
for both the Clifford and non-Clifford circuit, leading to negligible differences in the final
logical error-rate. A similar conclusion was made for the color-surface code cultivation protocol
including escape (with closed-boundary postselection) in recent work in Ref.
\cite{chase_clifft_2026}.

It remains to be understood whether further discrepancies arise between $\ket{\bar Y}$ and $\ket{\bar T}$
cultivation once the effects of open-boundary escape-stage postselection are included. It is
possible that an exact simulation of $\ket{\bar T}$ cultivation with our open-boundary postselection
method will reveal additional spacetime costs to recovering the performance of magic-state
cultivation, and is thus an important direction for future work.

\section{Adaptive Escape and Postselection Protocols}
\label{app:adaptive-escape}

The escape-stage optimizations considered above use a fixed code distance and a fixed number of syndrome-extraction rounds for every attempt. However, the information needed to make a reliable postselection decision can vary substantially between syndrome histories. Some attempts may already be accepted or rejected after a short escape stage, whereas others remain ambiguous and could benefit from further measurements. This motivates an adaptive escape protocol, in which the observed syndrome determines both the postselection decision and the resources allocated before that decision. The objective is to reduce the expected space-time volume per accepted magic state while maintaining a prescribed logical error rate.

\subsection{Syndrome-dependent termination}

Consider a sequence of decision points during escape. Let $h_k$ denote the information available at decision point $k$, including the observed syndrome, the code geometry, and previous adaptive operations. Define the adaptive counterpart of the open-boundary conditional logical error probability $\tilde P_L(\sigma_v)$ as
\begin{equation}
    q_k(h_k)
    =
    \Pr\!\left(
        \text{logical failure if accepted at }k
        \,\middle|\, h_k
    \right).
\end{equation}
Similar to the fixed case, this is inferred using only information available at decision time.

At each decision point, the protocol can accept, reject, or continue collecting syndrome information. A simple policy accepts when the estimated conditional logical error probability falls below a prescribed threshold. Otherwise, it continues if additional measurements are expected to justify their cost, and rejects if further measurements are unlikely to help or a maximum resource budget has been reached.

The distinction between continuation and rejection is essential. An inconclusive syndrome may indicate that the measurement record is too short to distinguish competing fault configurations, rather than that the attempt should be discarded. Additional measurements can resolve this ambiguity and allow some attempts to be accepted. Conversely, attempts that can already be accepted or rejected need not incur the cost of the longest escape stage.

\subsection{Temporal and spatial adaptation}

The simplest implementation holds the escape-code distance fixed and adapts the number of syndrome-extraction rounds. After an initial buffer of $r_{\min}$ rounds, the protocol evaluates an open-boundary postselection metric. Attempts that remain ambiguous undergo another round, or a small batch of rounds, before reevaluation. This continues until acceptance, rejection, or a maximum duration $r_{\max}$.

Additional rounds move the open temporal boundary away from earlier portions of the escape circuit and can make previously ambiguous fault configurations more distinguishable. However, longer evolution does not guarantee a lower conditional logical error probability: additional rounds introduce faults, and an open temporal boundary remains at every decision point. The benefit of further information must therefore be weighed against its resource cost and the additional logical error accumulated during continuation.

The same principle extends to the spatial direction. Instead of growing every attempt immediately to a common final distance, the protocol could use intermediate code sizes, similar to the original cultivation protocol but now applied to the escape stage. At each size, the syndrome would determine whether to accept, reject, perform additional rounds, or grow the code further. It is worth noting, however, that for the spatial case, the smaller patch sizes could lead to logical error contributions that are not reversible upon further code growth. Our fixed-geometry results motivate considering both controls, since increasing the escape duration alone need not recover the performance obtained by increasing both duration and distance.

\subsection{Resource objective}

Let $\pi$ denote an adaptive policy, $A$ the acceptance event, and $\tau$ the terminal decision point. Let $V_\pi$ be the volume consumed by one attempt, including injection, cultivation, and escape:
\begin{equation}
    V_\pi = \sum_{t=1}^{T_\pi} n_t\,\Delta t,
\end{equation}
where $T_\pi$ is the number of circuit layers, $n_t$ is the number of occupied physical qubits in layer $t$, and $\Delta t$ is the duration of each layer. Rejected attempts contribute the volume consumed before rejection. For independently repeated attempts with nonzero acceptance probability, the expected volume per accepted state is
\begin{equation}
    \overline V_\pi
    =
    \frac{\mathbb E_\pi[V_\pi]}{\Pr_\pi(A)}.
\end{equation}
The accepted logical error rate is
\begin{equation}
    \epsilon_\pi
    =
    \frac{
        \mathbb E_\pi\!\left[
            \mathbf 1_A q_\tau(h_\tau)
        \right]
    }{
        \Pr_\pi(A)
    }.
\end{equation}
The optimization objective is therefore
\begin{equation}
    \min_\pi \overline V_\pi
    \qquad \text{subject to} \qquad
    \epsilon_\pi \leq \epsilon_{\mathrm{target}}.
\end{equation}

Adaptation can reduce the volume spent on attempts that are already reliable, increase acceptance by resolving ambiguous histories, and avoid further expenditure on attempts that are unlikely to become acceptable.

With exact conditional probabilities, requiring
$q_\tau(h_\tau)\leq\epsilon_{\mathrm{target}}$ for every accepted history is sufficient to satisfy the error constraint. This is stronger than necessary, since the constraint applies to the average error probability of the accepted ensemble. More general policies could use different thresholds at different decision points while maintaining the same ensemble-level target.

Since adaptive postselection requires making a decision at an open boundary, we can use Caliper to supply a postselection score at each decision point, and decide on acceptance, rejection or continuation. Calibration is particularly important for adaptive decisions, as the raw Caliper estimate need not equal the conditional logical error probability, and attempts reaching later checkpoints have already been selected by earlier decisions. In addition, adaptive postselection may place stringent requirements on the classical processing time, and if classical processing requires additional waiting rounds, their volume and faults must be included in the resource estimate. We leave the detailed optimization of such an adaptive protocol to future work.

\end{document}